\documentclass[letterpaper,11pt]{article}
\usepackage[utf8]{inputenc}
\usepackage{microtype}
\usepackage[letterpaper,margin=1in]{geometry}
\usepackage[numbers,sort&compress]{natbib}
\usepackage{xcolor}
\usepackage{amsmath}
\usepackage{amssymb}
\usepackage{amsthm}
\usepackage{enumitem}
\usepackage{textgreek}
\usepackage{graphicx}
\usepackage{framed}
\usepackage[small]{caption}
\usepackage{booktabs}
\usepackage[unicode]{hyperref}
\usepackage{doi}

\theoremstyle{plain}
\newtheorem{theorem}{Theorem}
\newtheorem{lemma}[theorem]{Lemma}

\newtheorem*{conjecture*}{Conjecture}

\theoremstyle{definition}
\newtheorem{definition}[theorem]{Definition}

\theoremstyle{remark}

\newtheorem*{remark*}{Remark}

\newcommand{\namedref}[2]{\hyperref[#2]{#1~\ref*{#2}}}

\newcommand{\figureref}[1]{\namedref{Figure}{#1}}

\newcommand{\parent}{\mathsf {Parent}}
\newcommand{\lleft}{\mathsf {Left}}
\newcommand{\lright}{\mathsf {Right}}
\newcommand{\lchild}{\mathsf {LChild}}
\newcommand{\rchild}{\mathsf {RChild}}
\newcommand{\port}{\mathsf {Port}}
\newcommand{\noport}{\mathsf {NoPort}}
\newcommand{\portedge}{\mathsf {PortEdge}}
\newcommand{\gadedge}{\mathsf {GadEdge}}

\newcommand{\lindex}{\mathsf {Index}}
\newcommand{\lcenter}{\mathsf {Center}}
\newcommand{\lerror}{\mathsf {Error}}
\newcommand{\lup}{\mathsf {Up}}
\newcommand{\ldown}{\mathsf {Down}}

\newcommand{\llout}{\mathsf {out}}
\newcommand{\lin}{\mathsf {in}}
\newcommand{\NN}{\mathbb {N}}
\newcommand{\AAA}{\mathcal {A}}
\newcommand{\gadget}{\mathcal {G}}
\newcommand{\gadrepr}{{\hat{G}}}
\newcommand{\gadcheck}{\mathcal {V}}

\newcommand{\lgadok}{\mathsf {GadOk}}

\newcommand{\lporterr}{\mathsf {PortErr}}
\newcommand{\lnoporterr}{\mathsf {NoPortErr}}
\newcommand{\lok}{\mathsf {Ok}}
\newcommand{\lerr}{\mathsf {Err}}

\newcommand{\inn}{{\operatorname{in}}}
\newcommand{\out}{{\operatorname{out}}}
\newcommand{\llist}{{\operatorname{list}}}

\newcommand{\tdet}{T_{\operatorname{det}}}
\newcommand{\trand}{T_{\operatorname{rand}}}
\newcommand{\netdec}{\mathsf{ND}}

\newcommand{\lcl}{{\upshape\sffamily LCL}}
\newcommand{\nelcl}{{\upshape\sffamily ne-LCL}}
\newcommand{\local}{{\upshape\sffamily LOCAL}}

\DeclareMathOperator{\poly}{poly}

\newenvironment{myabstract}
{\list{}{\listparindent 1.5em%
        \itemindent    \listparindent
        \leftmargin    1cm
        \rightmargin   1cm
        \parsep        0pt}%
    \item\relax}
{\endlist}

\newenvironment{mycover}
{\list{}{\listparindent 0pt
        \itemindent    \listparindent
        \leftmargin    1cm
        \rightmargin   1cm
        \parsep        0pt}%
    \raggedright
    \item\relax}
{\endlist}

\newcommand{\myemail}[1]{\,$\cdot$\, {\small #1}}
\newcommand{\myaff}[1]{\,$\cdot$\, {\small #1}\par\smallskip}

\hypersetup{
    colorlinks=true,
    linkcolor=black,
    citecolor=black,
    filecolor=black,
    urlcolor=[HTML]{0088CC},
    pdftitle={How much does randomness help with locally checkable problems?},
    pdfauthor={Alkida Balliu, Sebastian Brandt, Dennis Olivetti, Jukka Suomela}
}

\begin{document}

\begin{mycover}
    {\huge\bfseries How much does randomness help with locally checkable problems? \par}
    \bigskip
    \bigskip

    \textbf{Alkida Balliu}
    \myemail{alkida.balliu@aalto.fi}
    \myaff{Aalto University}

    \textbf{Sebastian Brandt}
    \myemail{brandts@ethz.ch}
    \myaff{ETH Zurich}

    \textbf{Dennis Olivetti}
    \myemail{dennis.olivetti@aalto.fi}
    \myaff{Aalto University}

    \textbf{Jukka Suomela}
    \myemail{jukka.suomela@aalto.fi}
    \myaff{Aalto University}
\end{mycover}
\bigskip

\begin{myabstract}
\noindent\textbf{Abstract.}
Locally checkable labeling problems (\lcl{}s) are distributed graph problems in which a solution is globally feasible if it is locally feasible in all constant-radius neighborhoods. Vertex colorings, maximal independent sets, and maximal matchings are examples of \lcl{}s.

On the one hand, it is known that some \lcl{}s benefit \emph{exponentially} from randomness---for example, any deterministic distributed algorithm that finds a \emph{sinkless orientation} requires $\Theta(\log n)$ rounds in the \local{} model, while the randomized complexity of the problem is $\Theta(\log \log n)$ rounds. On the other hand, there are also many \lcl{}s in which randomness is useless.

Previously, it was not known if there are any \lcl{}s that benefit from randomness, but only \emph{subexponentially}. We show that such problems exist: for example, there is an \lcl{} with deterministic complexity $\Theta(\log^2 n)$ rounds and randomized complexity $\Theta(\log n \log \log n)$ rounds.
\end{myabstract}

\thispagestyle{empty}
\setcounter{page}{0}
\newpage

\section{Introduction}

\paragraph{Locality of locally checkable problems.}

One of the big themes in the theory of distributed graph algorithms is \emph{locality}: given a graph problem, how \emph{far} does an individual node need to see in order to be able to produce its own part of the solution? This idea is formalized as the time complexity in the \local{} model \cite{Linial1992,Peleg2000} of distributed computing.

\begin{figure}
\centering
\includegraphics[page=1,width=\textwidth]{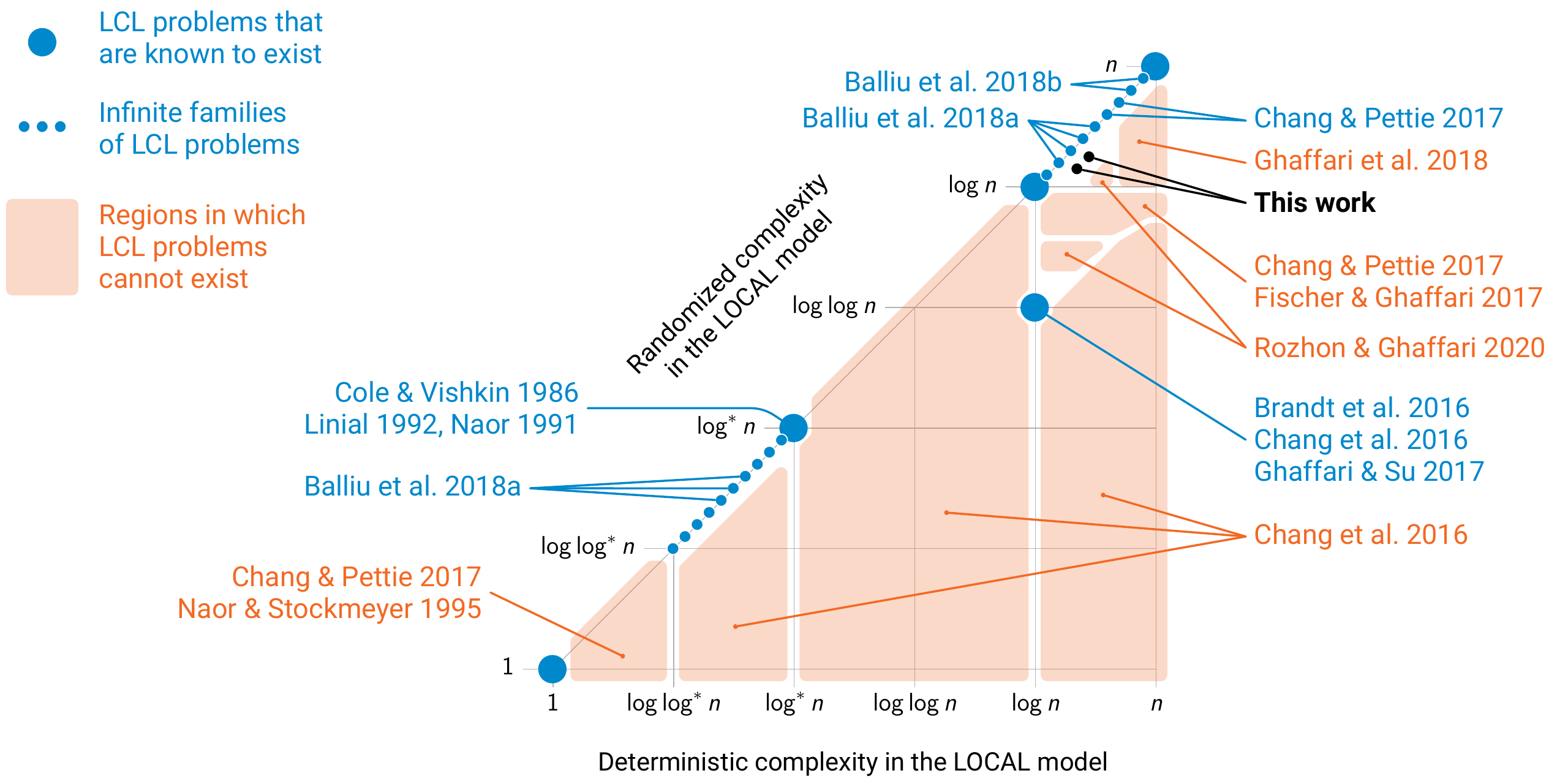}
\caption{The landscape of \lcl{} problems in the \local{} model: possible complexities (blue dots) and gaps (orange shading), based on the papers \cite{cole86deterministic,Linial1992,Naor1991,Balliu2018stoc,Balliu2018disc,Chang2019,Ghaffari2018,fischer17sublogarithmic,Rozhon2019,Brandt2016,chang16exponential,ghaffari17distributed,Naor1995}. The present work gives the last major piece of the puzzle: there are indeed problems in which randomness helps, but only polynomially.}\label{fig:landscape}
\end{figure}

While we are still very far from understanding the locality of all possible graph problems, there is one highly relevant family of graph problems  that is now close to being completely characterized: \emph{locally checkable labeling problems}, or in brief \lcl{}s. In essence, \lcl{}s are graph problems in which feasible solutions are easy to verify in a distributed setting---if a solution looks good in all local neighborhoods, it is also good globally. This family of problems was introduced in the seminal paper by \citet{Naor1995} in the 1990s, and while the groundwork for understanding the locality of \lcl{}s was done already in the 1980s--1990s \cite{Goldberg1988,cole86deterministic,panconesi95delta,Linial1992,Naor1991}, most of the progress is from the past four years \cite{Balliu2018stoc,Balliu2018disc,Brandt2016,chang16exponential,Chang2019,fischer17sublogarithmic,ghaffari17distributed,Ghaffari2018,Ghaffari2018a,chang18complexity,Brandt2017}.

There are many relevant graph classes to study, but for our purposes the most interesting case is \emph{general bounded-degree graphs}. We only assume that there is some constant upper bound $\Delta$ on the maximum degree of the graph, and other than that there is no promise about the structure of the input graph. If there are $n$ nodes, the nodes will have unique identifiers from $\{1,2,\dotsc,\poly(n)\}$, and initially each node knows $n$, $\Delta$, its own identifiers, and its own degree---everything else it has to learn through communication.

For bounded-degree graphs, the state of the art is summarized in \figureref{fig:landscape}. The figure represents the landscape of all possible distributed time complexities of \lcl{} problems, both for deterministic and randomized algorithms. There are infinite families of problems with distinct time complexities, but there are also large \emph{gaps}: for example, \citet{chang16exponential} showed that there is no \lcl{} with a time complexity in the range $\omega(\log^* n)$ and $o(\log n)$. For deterministic algorithms, the work of characterizing possible time complexities of \lcl{} problems is near-complete.

\paragraph{Role of randomness.}

What we aim at understanding is \emph{how much randomness helps} with \lcl{}s. As shown in \figureref{fig:landscape}, there are some problems in which randomness helps \emph{exponentially}. The most prominent example is \emph{sinkless orientation}: its deterministic complexity is $\Theta(\log n)$, while the randomized complexity is $\Theta(\log\log n)$ \cite{Brandt2016,chang16exponential,ghaffari17distributed}.

On the one hand, it is known that there is at most an exponential gap between deterministic and randomized complexities \cite{chang16exponential}. On the other hand, there are also lower bounds that exclude many possible combinations of deterministic and randomized time complexities. As illustrated in \figureref{fig:landscape}, the work by \citet{Chang2019} and \citet{fischer17sublogarithmic} implies that there is no \lcl{} with deterministic complexity $\Theta(\log n)$ and randomized complexity e.g.\ $\Theta(\sqrt{\log n})$. If a problem can be solved in deterministic logarithmic time, then either randomness helps a lot or not at all.

Sinkless orientation and closely related problems such as \emph{$\Delta$-coloring} and \emph{algorithmic Lov\'asz local lemma} are currently the only \lcl{}s for which randomness is known to help. Indeed, all known results previous to our work are compatible with the following conjecture:

\begin{oframed}
\begin{conjecture*}
If the deterministic complexity of an \lcl{} is $\Theta(\log n)$, then its randomized complexity is either $\Theta(\log n)$ or $\Theta(\log \log n)$. Otherwise the randomized complexity is asymptotically equal to the deterministic complexity.

In particular, randomness helps exponentially or not at all.
\end{conjecture*}
\end{oframed}

We show that the conjecture is false. We show that there are \lcl{} problems that benefit from randomness, but only \emph{polynomially}. We show how to construct, e.g., an \lcl{} with deterministic complexity $\Theta(\log^2 n)$ rounds and randomized complexity $\Theta(\log n \log \log n)$ rounds. 

The role of randomness in distributed computing is a key research question, and it has been extensively studied. In fact, in their book on graph coloring, \citet{Barenboim2013} write ``Perhaps the most fundamental open problem in this field is to understand the power and limitations of randomization.'' We make a step forward in the understanding of this fundamental question.

\paragraph{Technique: padding.}

\begin{figure}[t]
    \centering
    \includegraphics[scale=0.5]{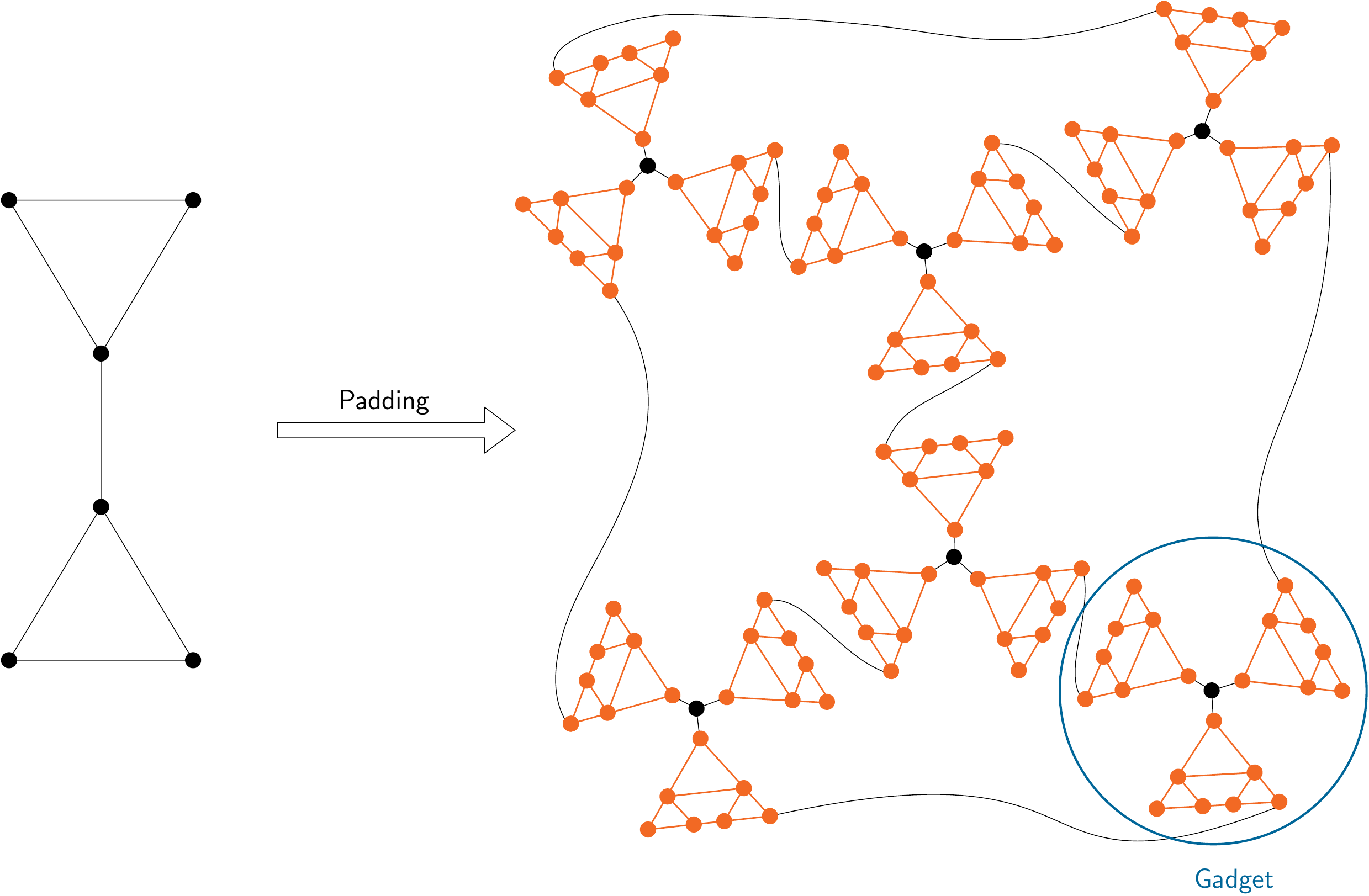}
    \caption{Padding: replacing each node of $G$ with a tree-like gadget from family $\gadget$ to obtain a new graph $G'$.}\label{fig:padding}
\end{figure}

The main technical idea is to introduce the concept of \emph{padding} in the construction of \lcl{} problems---the basic idea is inspired by the padding technique in the classical computational complexity theory \cite[Sect.~2.6]{Arora2009}.

We start with an \lcl{} problem $\Pi$ and a suitable family of \emph{gadgets} $\gadget$. Then we use the gadgets to construct a new graph problem $\Pi'$ such that both deterministic and randomized complexity of $\Pi'$ is higher than those of $\Pi$. More concretely, let $\Pi$ be the problem of finding a sinkless orientation, with randomized complexity $\Theta(\log \log n)$ and deterministic complexity $\Theta(\log n)$, and let $\gadget$ be a suitable family of tree-like graphs. By applying $\gadget$ to $\Pi$, we obtain $\Pi'$ in which both randomized and deterministic time complexity have increased by a factor of $\Theta(\log n)$; hence the randomized complexity of $\Pi'$ is $\Theta(\log n \log \log n)$ and the deterministic complexity is $\Theta(\log^2 n)$. By applying $\gadget$ to $\Pi'$ recursively, we can then further obtain randomized complexity $\Theta(\log^i n \log\log n)$ and deterministic complexity $\Theta(\log^{i+1}n)$ for any constant $i$.

\figureref{fig:padding} shows what we would ideally like to do: given a hard instance $G$ for $\Pi$, we replace each node with a suitable gadget to obtain a hard instance $G'$ for $\Pi'$. The intuition here is that padding increases distances, so if all gadgets happened to be trees of depth $x$, then solving $\Pi'$ on $G'$ is exactly $x$ times as hard as solving $\Pi$ on $G$.

This would be easy to implement if we had a \emph{promise} that the input is of a suitable form, but with a promise one can trivially construct \lcl{}s with virtually any complexity. The key challenge is implementing the idea so that $\Pi'$ is an \lcl{} in the strict sense and we can control its distributed time complexity also in the family of all bounded-degree graphs. Some challenges we need to address include:
\begin{enumerate}[noitemsep]
    \item What if we have an input graph $G'$ that is not of the right form, i.e., it does not consist of valid gadgets connected to each other?
    \item What if we have an input graph $G'$ in which the gadgets have different depths?
\end{enumerate}

The first challenge we overcome by making the gadgets locally checkable. In essence, a node will be able to see within distance $O(\log n)$ if it is part of an invalid gadget, and it is also able to construct a \emph{locally checkable proof of error}. \lcl{} $\Pi'$ is defined so that we have to either solve the original problem $\Pi$ or produce locally checkable proofs of errors. This ensures that:
\begin{itemize}[noitemsep]
    \item An algorithm solving $\Pi'$ cannot cheat and claim that the input is invalid if this is not the case.
    \item The adversary who constructs input $G'$ never benefits from a construction that contains invalid gadgets, as they will in essence result in ``don't care'' nodes that only make solving $\Pi'$ easier.
\end{itemize}
See e.g.\ \cite{Goos2016,korman10proof} for more details on the concept of locally checkable proofs; in our case it will be essential that errors have a locally checkable proof with constantly many bits per node so that we can interpret it as an \lcl{}.

The second challenge we overcome by choosing the original problem $\Pi$ and the gadget family $\gadget$ so that the worst case input that the adversary can construct is essentially of the following form:
\begin{itemize}[noitemsep]
    \item Start with an $n$-node graph $G$ that is a worst-case input for $\Pi$.
    \item Replace each node with an $n$-sized gadget, which has depth $\Theta(\log n)$.
\end{itemize}
This way in the worst case the adversary can construct a graph $G'$ with $N = n^2$ nodes, and if solving $\Pi$ on $G$ took $\Theta(\log \log n)$ rounds for randomized algorithms and $\Theta(\log n)$ rounds for deterministic algorithms, then solving $\Pi'$ on $G'$ will take $\Theta(\log n \log \log n) = \Theta(\log N \log \log N)$ rounds for randomized algorithms and $\Theta(\log^2 n) = \Theta(\log^2 N)$ rounds for deterministic algorithms. We can show that a different balance between the size of $G$ and the depth of each gadget will not result in a harder instance; both much larger and much smaller gadgets will only make the problem easier.

\paragraph{Discussion and open questions.}

If we write $D(n)$ for the deterministic complexity and $R(n)$ for the randomized complexity of a given \lcl{}, we have now seen that we can engineer \lcl{}s that satisfy e.g.\ any of the following:
\begin{align*}
    R(n) &\approx D(n), \\
    R(n) &\approx \sqrt{D(n)}, \\
    R(n) &\approx \log D(n).
\end{align*}
However, if we look at the ratio $D(n)/R(n)$, we see that all examples with $R(n) = o(D(n))$ happen to satisfy
\[
    \frac{D(n)}{R(n)} = \Theta\biggl(\frac{\log n}{\log \log n}\biggr).
\]
The main open question is whether we can construct \lcl{}s with
\[
    \frac{D(n)}{R(n)} \gg \log n.
\]
This question is closely connected to the complexity of \emph{network decompositions}: the result of \citet{Ghaffari2018} implies that, in the context of \lcl{}s, any randomized algorithm running in time $R(n)$ can be transformed to a deterministic algorithm running in time $D(n) = O\bigl( R(n) \netdec(n) + R(n) \log^2 n \bigr)$, where $\netdec(n)$ is the time required to compute a $(\log n, \log n)$-network decomposition in graphs of size $n$ with a deterministic distributed algorithm. A recent breakthrough by \citet{Rozhon2019} showed that the network decomposition problem can be solved in polylogarithmic time. In particular, they provided an algorithm running in $O(\log^7 n)$ rounds in the \local{} model. This implies that the ratio $D(n)/R(n)$ cannot be more than polylogarithmic. However, whether the ratio can be superlogarithmic is an open question: if we could improve our result slightly and show the existence of \lcl{}s satisfying $D(n)/R(n) = \omega(\log^2 n)$, we would obtain a superlogarithmic lower bounds for the network decomposition problem---a long-standing open question.

\section{Preliminaries}

\paragraph{Model.}
The \local{} model is synchronous, that is, the computation proceeds in synchronous rounds. At each round, each entity sends messages to its neighbors, receives messages from them, and performs some computation based on the data it receives. In this model, the size of messages can be arbitrarily large, and the computational power of an entity is not bounded. The time complexity $T$ of an algorithm running in the \local{} model is given by the number of rounds that entities need to run the algorithm in order to solve a problem.

The \local{} model is equivalent to a model where each entity:
(i)~gathers its radius-$T$ neighborhood, i.e., the entity learns the structure of the network around it up to distance $T$, along with the inputs that the entities in this neighborhood might have;
(ii)~performs some computation based on the data that has been gathered;
(iii)~produces its own local output.

A distributed network is represented by a graph with nodes and edges, where a node represents a specific entity of the network, and there is an edge between two nodes if and only if there is a communication link between the entities that they represent. We denote a graph by $G=(V,E)$, where $V$ is the set of nodes and $E$ the set of edges. The degree $d$ of a node is the number of its incident edges. The incident edges are numbered, that is, we assume that a node has \emph{ports} numbered from $1$ to $d$ where incident edges are connected to. Each node, when receiving a message, knows the port from which the message arrives. We denote by $\Delta$ the maximum degree in the graph. 

For technical reasons, we deviate from the usual assumptions and we allow $G$ to be disconnected and to contain self loops and parallel edges. While all upper and lower bounds that we will present hold in this larger class of graphs, our final results hold for simple graphs as well.

\paragraph{Locally checkable labeling problems.}
\lcl{} problems are defined on constant degree graphs, i.e., graphs where $\Delta = O(1)$. Each node has an input label from a constant-size set $\Sigma_\inn$, and must produce an output label from a constant-size set $\Sigma_\out$. The output must be locally checkable, that is, there must exist a constant-time distributed algorithm that can check the correctness of a solution. If the solution is globally correct, this algorithm must accept on all nodes, otherwise it must reject on at least one node. A distributed algorithm $\mathcal{A}$ solving an \lcl{} problem in time $T(n)$ is an algorithm that, for any graph $G$ with $n$ nodes, given $n$ and $\Delta$, runs in $T(n)$ rounds and outputs a label for each node, such that the \lcl{} constraints are satisfied at each node. For randomized algorithms, we require global high probability of success, that is, the probability that the solution is wrong must be at most $\frac{1}{n}$.

An example of an \lcl{} problem is the proper ($\Delta +1$)-coloring of the nodes of a graph: nodes have all the same input, that is a special character denoting the empty input label, and they must produce as output a color in $\{1,\dots,\Delta+1\}$. In a proper coloring it must hold that, for any pair of neighbors, their colors are different. It is easy to see that, if the graph is properly colored, each node will see a proper solution locally, otherwise there will be two neighbor nodes that will have the same color, noticing the error. Many other natural problems fall in the category of \lcl{}s, such as edge coloring, maximal matching, maximal independent set, sinkless orientation, etc. 

Deviating from the common way of writing inputs and outputs of \lcl{}s only on nodes (or, occasionally, edges), we will write inputs and outputs on nodes, edges, and node-edge pairs.
This allows us to conveniently assign different labels to each half of an edge, something we will make use of in Section~\ref{sec:thegadget}.
For technical reasons we restrict our considerations to the subclass of \lcl{}s where the local constraints determining whether a solution is correct can be checked ``on nodes and edges".
Note that almost all commonly studied \lcl{} problems can be reformulated in this form, by requiring each node to return, apart from its own output, also the outputs of all nodes at a constant distance.
Formally, these \emph{node-edge-checkable} \lcl{}s, or \nelcl{}s, are defined as follows.

\paragraph{Node-edge-checkable \lcl{}s.}
Let $B = \{ (v,e) \in V \times E \mid v \in e \}$ be the set of incident node-edge pairs.
The input to an \nelcl{} is given by assigning an input label $i \in \Sigma_{\inn}$ to each $x \in V \cup E \cup B$; a solution to an \nelcl{} is given by each node $v$ assigning an output label $o \in \Sigma_{\out}$ to itself and to each ``incident" element of $V \cup E \cup B$, where for each edge $e = \{ u, v \}$, nodes $u$ and $v$ have to choose the same output label for $e$.
Apart from the sets $\Sigma_{\inn}$ and $\Sigma_{\out}$ of input and output labels, an \nelcl{} is defined by a set $C_N$ of node constraints and a set $C_E$ of edge constraints, where $C_N$ describes for each node $v$ which output label configurations on $\{ v \} \cup \{ \{ v, u \} \in E \} \cup \{ (v, e) \in B \}$ are correct (depending on the input labels on those nodes, edges, and node-edge pairs), and $C_E$ describes for each edge $e = \{ u, v \}$ which output label configurations on $\{ u, v, e, (u, e), (v, e) \}$ are correct (again, depending on the input labels on those elements of $V \cup E \cup B$).
Note that $C_N$ and $C_E$ do not depend on the choice of $v$ or $e$ in the above description, or on the port numbers or identifiers assigned to the edges or nodes of the graph.

As an example, let us see how sinkless orientation can be formulated as an \nelcl{}. Each node $v$ has to output on each incident edge, or more precisely on each $(v,e) \in B$, either the label $\llout$ (outgoing) or the label $\lin$ (incoming). The constraint on nodes is that there must exist an incident edge labeled $\llout$. This guarantees that no node is a sink. The constraint on edges is that, whenever an endpoint is labeled $\lin$, the other endpoint must be labeled $\llout$, and vice versa. This guarantees that the edges are oriented consistently. Note that in this example, the constraints are independent of any input labels. See Figure \ref{fig:node-edge-labels} for an illustration.

\begin{figure}[t]
	\centering
	\includegraphics[scale=0.8]{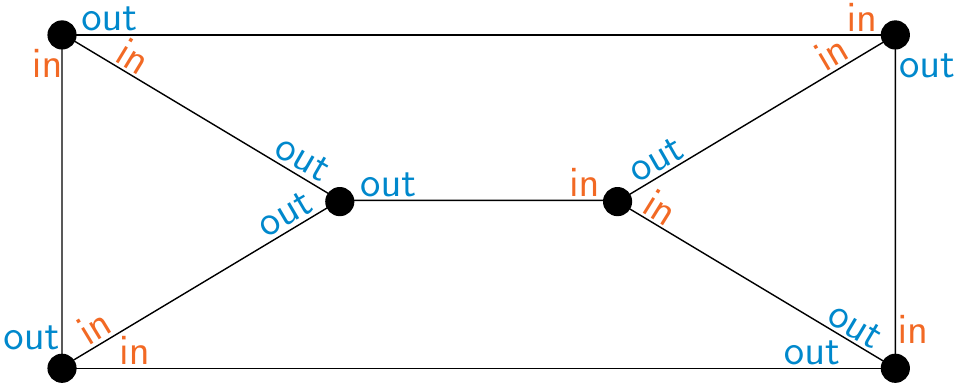}
	\caption{Sinkless orientation in the node-edge pair formalism.}\label{fig:node-edge-labels}
\end{figure}

\section{Padded \lcl{}s}\label{sec:padded-lcl}
In this section we provide a technique that constructs new \lcl{}s in a black box manner.
More precisely, given an \nelcl{} $\Pi$ and a collection of graphs, so-called gadgets, with certain properties, we can construct a new \nelcl{} $\Pi'$ with changed deterministic and randomized complexities.
Informally, the idea is that the hard graphs for $\Pi'$ are so-called padded graphs, i.e., graphs obtained by taking some graph $G$ and replacing each node of $G$ with a gadget, thereby ``padding" $G$.
See Figure \ref{fig:padding} for an example.

Our new \nelcl{} $\Pi'$ is constructed in a way that ensures that in such a padded graph solving $\Pi'$ is equivalent to solving $\Pi$ on the underlying initial graph $G$.
Moreover, the padding itself will make sure that the distances between nodes of $G$ increase; in other words, simulating an algorithm on $G$ that solves $\Pi$ incurs an additional communication overhead.
Consequently the hard graphs for $\Pi'$ are given by those instances where the underlying graph $G$ belongs to the hard graphs for $\Pi$ and the size of the gadgets used in the padding is finely balanced such that (1)~the underlying graph $G$ is large enough (as a function of the number $n$ of nodes of the padded graph) to ensure a sufficiently large runtime for solving $\Pi$ on $G$, and (2)~the gadgets are large enough to ensure a sufficiently large communication overhead.

We will start the section by defining gadgets and families thereof; in particular, we will describe their special properties that will enable us to define $\Pi'$ and prove that it has the desired complexities.
Then we will give a formal definition of padded graphs which, intuitively, are the key concept for the subsequent definition of the new \nelcl{} $\Pi'$, even though, formally, they do not appear in the definition.
After defining $\Pi'$, we will conclude the section by showing how the complexity of the new \nelcl{} $\Pi'$ is related to the complexity of the old \nelcl{} $\Pi$.

The exact relation between the complexities of the two \nelcl{}s (which relies on the subsequently defined concept of a $(d, \Delta)$-gadget family) is given in Theorem~\ref{thm:pi_to_newpi}. 
Let $\tdet(\Pi, N)$, resp.\ $\trand(\Pi, N)$, denote the deterministic, resp.\ randomized, complexity of an \lcl{} $\Pi$ on instances of size $N$.
Then the following holds.
\begin{theorem}\label{thm:pi_to_newpi}
	Let $f\colon \NN \to \NN$ be a function such that, for each $x \in \NN$, we have $f(x) \leq x$  and there exists some $y \in \NN$ with $f(y) = x$.
	For each \nelcl{} problem $\Pi$ and each $(d, \Delta)$-gadget family $\gadget$, there exists an \nelcl{} problem $\Pi'$ with deterministic complexity $O\bigl(\tdet(\Pi,n) \cdot d(n)\bigr)$ and $\Omega\bigl(\tdet(\Pi,f(n)) \cdot d(\frac{n}{f(n)})\bigr)$ and randomized complexity $O\bigl(\trand(\Pi,n) \cdot d(n)\bigr)$ and $\Omega\bigl(\trand(\Pi,f(n)) \cdot d(\frac{n}{f(n)})\bigr)$.
\end{theorem} 

\subsection{Gadgets}\label{sec:gadgets}
\begin{definition}\label{def:gadget}
	An \emph{$(n,D)_\Delta$-gadget} $F$ is a (labeled) connected graph that satisfies the following:
	\begin{itemize}
		\item The number of nodes is $n$.
		\item There are exactly $\Delta$ special nodes labeled $\port_i$, for $1 \le i \le \Delta$, called ports. All other nodes are labeled $\noport$.
		\item The diameter of $F$ and hence also the pairwise distances between the ports are at most $D$.
	\end{itemize}

	Let $d\colon \NN \to \NN$ be some function.
	A \emph{$(d, \Delta)$-gadget family} $\gadget$ is a set of graphs satisfying the following:
	\begin{itemize}
		\item Each $G \in \gadget$ is an $(n,O(d(n)))_\Delta$-gadget for some $n$.
		\item For each $n \in \NN$, there exists some $G \in \gadget$ with $\Theta(n)$ nodes such that the pairwise distances between the ports are all in $\Theta(d(n))$. Let this gadget be $\gadrepr_n$.
		\item There is an \nelcl{} $\Psi_{\gadget}$ with the following properties, where $H$ denotes the input graph for $\Psi_{\gadget}$.
		\begin{itemize}
			\item The output label set for $\Psi_{\gadget}$ is $\{ \lgadok \} \mathbin{\dot\cup} L_{\lerr}$, for some finite set $L_{\lerr}$.
			\item If $H \in \gadget$, then the unique (globally) correct solution for $\Psi_{\gadget}$ uses only the output label $\lgadok$.
			\item If $H \notin \gadget$, then there exists a (globally) correct solution for $\Psi_{\gadget}$ that uses only output labels from $L_{\lerr}$.
			\item There is a deterministic distributed algorithm $\gadcheck$ that, given an upper bound $n$ of $N$, where $N$ is the number of nodes of $H$, solves $\Psi_{\gadget}$ in $O(d(n))$ rounds. Moreover, if $H \notin \gadget$, then $\gadcheck$ uses only output labels from $L_{\lerr}$. We call the (global) output of $\gadcheck$ a \emph{locally checkable proof of error}.
		\end{itemize}
	\end{itemize}
\end{definition}

\subsection{Padded graphs}\label{sec:padd}
Intuitively, a padded graph is a graph obtained by starting from some arbitrary graph and replacing each node with a gadget $F \in \gadget$. We now formally define the family $\gadget(G)$ of padded graphs for a given graph $G$.

\begin{definition}
	Given a graph $G$ with maximum degree $\Delta$ and a $(d, \Delta)$-gadget family $\gadget$, the graph family $\gadget(G)$ is the set of all graphs that can be obtained by the following process.
	
	Start from $G = (V,E)$. For each node $v \in V$ pick a gadget $F \in \gadget$, where different gadgets may be picked for different nodes.
	Let $C_v$ be the gadget chosen for node $v$.
	The final graph is the union of the $C_v$ (over all $v \in V$), augmented by the following additional edges: for any edge $\{u,v\} \in E$ connecting port $a$ of $u$ to port $b$ of $v$, add an edge between node $\port_a$ of $C_u$ and $\port_b$ of $C_v$.
	Moreover, in the final graph we label
	each edge already present in the union of the $C_v$ with $\gadedge$, and each edge that has been added in the augmentation step with $\portedge$.
\end{definition}

\subsection{New \lcl{}}
Given an \lcl{} $\Pi$ and a $(d, \Delta)$-gadget family $\gadget$, in this section we define a new \lcl{} $\Pi'$ that, informally, can be described as follows.
Each edge $e$ of the input graph $G$ for $\Pi'$ is assigned a special label that indicates whether $e$ belongs to a gadget or to ``the underlying graph", denoted by $H$.
Intuitively, $H$ is the graph obtained by contracting the connected components induced by the edges labeled as belonging to a gadget.
For each such connected component, there are two possibilities: Either it constitutes a gadget from our gadget family $\gadget$, in which case we call it a valid gadget, or it does not, in which case we call it an invalid gadget.

In each invalid gadget, $\Pi'$ can be solved correctly by the containing nodes providing a locally checkable proof of the invalidity of the gadget.
Consider the graph obtained by deleting all gadgets where the contained nodes proved an error.
Assuming that all invalid gadgets have been claimed to be invalid by their contained nodes (we do not require that nodes in an invalid gadget actually choose this option) and consequently deleted, the obtained graph $G'$ may still not be a padded graph as described in Section~\ref{sec:padd}. In fact, while padded graphs satisfy that a gadget $F$ corresponding to node $v$ of degree $d$ has nodes $\port_1,\ldots,\port_d$ connected to port nodes of other gadgets, $G'$ may have some port nodes connected to removed gadgets, thus valid port nodes are an arbitrary subset of $\{\port_i ~|~ 1 \le i \le \Delta \}$. This implies that we can transform $G'$ to a valid padded graph in a natural way, by just mapping the $d$ ($0\le d \le \Delta$) valid port nodes to the ports from $1$ to $d$. We will actually require nodes to produce such a mapping, and mark each port node as valid or invalid (see Figure \ref{fig:port-mapping} for an example).
Then, $\Pi'$ is solved correctly if the nodes solve $\Pi$ on the graph obtained from $G'$ by contracting all valid gadgets.

\begin{figure}[t]
	\centering
	\includegraphics[scale=0.7]{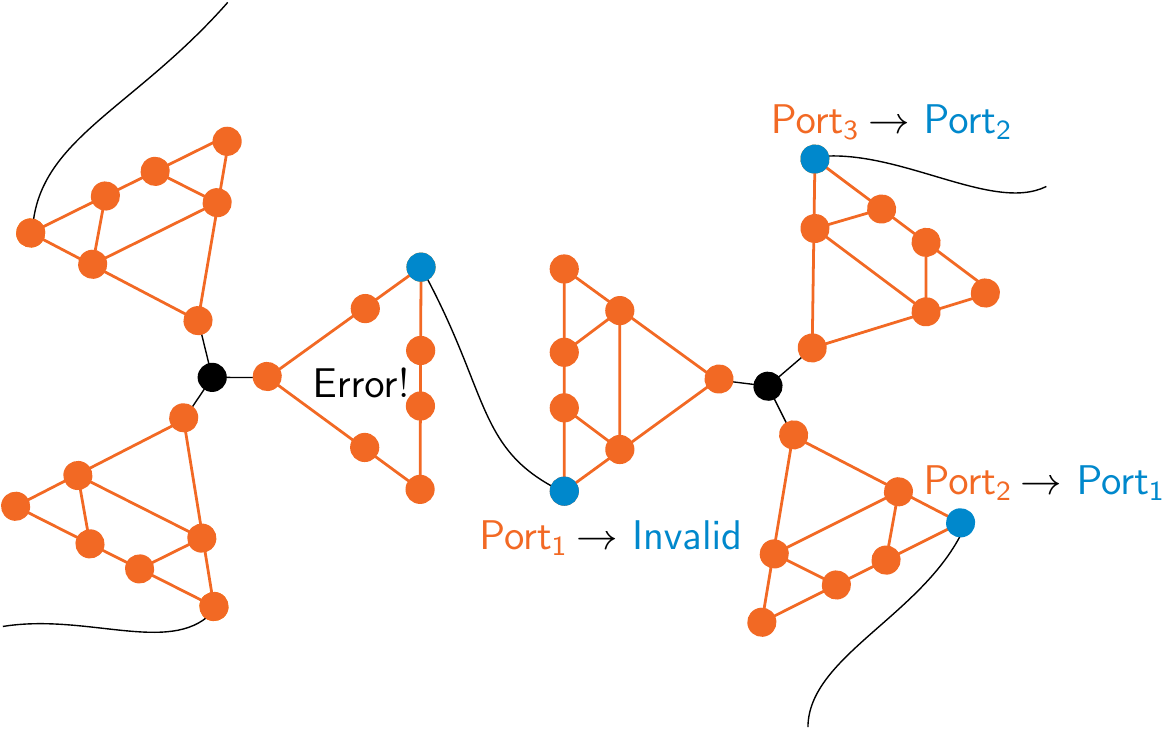}
	\caption{An example of port mapping: $\port_1$ is connected to an invalid gadget, while  $\port_2$ and  $\port_3$ are connected to valid gadgets; nodes of the valid gadget produce the mapping $2\rightarrow 1$, $3\rightarrow 2$.}\label{fig:port-mapping}
\end{figure}

Some care is needed to ensure that the above rough outline can be expressed in terms of \nelcl{} constraints and to deal with the subtleties introduced thereby.
We now proceed by defining $\Pi'$.

Let the \nelcl{} $\Pi$ be given by input label set $\Sigma^{\Pi}_{V,\inn} \cup \Sigma^{\Pi}_{E,\inn} \cup \Sigma^{\Pi}_{B,\inn}$, output label set $\Sigma^{\Pi}_{V,\out} \cup \Sigma^{\Pi}_{E,\out} \cup \Sigma^{\Pi}_{B,\out}$, node constraint set $C^{\Pi}_V$, and edge constraint set $C^{\Pi}_E$.
Let $\gadget$ be an arbitrary $(d, \Delta)$-gadget family and let $\Psi_{\gadget}$ be as described in Definition~\ref{def:gadget}.
Recall that the graphs in $\gadget$ are labeled.
Let \[\bigl(\Sigma^{\gadget}_{V,\inn} \times \{\port_1, \dots, \port_{\Delta}, \noport\}\bigr) \cup \Sigma^{\gadget}_{E,\inn} \cup \Sigma^{\gadget}_{B,\inn}\] denote the labels used for labeling the graphs in $\gadget$, which will be our input labels for $\Psi_{\gadget}$.
Let $\Sigma^{\gadget}_{V,\out} \cup \Sigma^{\gadget}_{E,\out} \cup \Sigma^{\gadget}_{B,\out}$, $C^{\gadget}_V$, and $C^{\gadget}_E$ denote the output labels, node constraints, and edge constraints for $\Psi_{\gadget}$, respectively.
In particular, we have \[\Sigma^{\gadget}_{V,\out} \cup \Sigma^{\gadget}_{E,\out} \cup \Sigma^{\gadget}_{B,\out} = \lgadok \mathbin{\dot\cup} L_{\lerr}.\]

W.l.o.g., we can (and will) assume that both in $\Pi$ and in $\Psi_{\gadget}$, each element of $V \times E \times B$ is assigned exactly one input label (and each will receive exactly one output label) as we can encode multiple labels in one label and add an ``empty label" for the case that no label was assigned.
However, for convenience, we might deviate from this underlying encoding in the description of the new \nelcl{} $\Pi'$.
We now give a formal definition of $\Pi'$. We will later provide an informal explanation of each part.

\paragraph{Input labels.}
\begin{itemize}[noitemsep]
	\item Each node has a label in $\Sigma^{\Pi}_{V,\inn} \times \Sigma^{\gadget}_{V,\inn} \times \{\port_1, \dots, \port_{\Delta}, \noport\}$.
	\item Each edge has a label in $\Sigma^{\Pi}_{E,\inn} \times \Sigma^{\gadget}_{E,\inn} \times \{\portedge, \gadedge\}$.
	\item Each element of $B$ has a label in $\Sigma^{\Pi}_{B,\inn} \times \Sigma^{\gadget}_{B,\inn}$.
\end{itemize}
\paragraph{Output labels.}
\begin{itemize}[noitemsep]
	\item Each node must label itself with a label from $\Sigma_{\llist} \times \{\lporterr_1, \lporterr_2, \lnoporterr\} \times \Sigma^{\gadget}_{V,\out}$, where
	\[
		\Sigma_{\llist} = 2^{\{\port_1, \dots, \port_{\Delta}\}} \times \Sigma^{\Pi}_{V,\inn}\times (\Sigma^{\Pi}_{E,\inn})^{\Delta} \times (\Sigma^{\Pi}_{B,\inn})^{\Delta} \times \Sigma^{\Pi}_{V,\out}\times (\Sigma^{\Pi}_{E,\out})^{\Delta} \times (\Sigma^{\Pi}_{B,\out})^{\Delta} \enspace.
	\]
	\item Each edge must be labeled with either $\epsilon$ or a label from $\Sigma^{\gadget}_{E,\out}$.
	\item Each element of $B$ must be labeled with either $\epsilon$ or a label from $\Sigma^{\gadget}_{B,\out}$.
\end{itemize}

\paragraph{Constraints.}
\begin{enumerate}
	\item\label{item:eps} Each edge with input label $\portedge$ has to be labeled $\epsilon$, each edge with input label $\gadedge$ has to be labeled with a label from $\Sigma^{\gadget}_{E,\out}$. Each $(v, e) \in B$ has to be labeled $\epsilon$ if $e$ has input label $\portedge$, and with a label from $\Sigma^{\gadget}_{B,\out}$ if $e$ has input label $\gadedge$.
	\item\label{item:check} On each connected component of the subgraph induced by the edges labeled $\gadedge$, the \nelcl{} $\Psi_{\gadget}$ has to be solved correctly. Put in a local way, for each node $v$ the node constraints $C^V_{\gadget}$ of $\Psi_{\gadget}$ have to be satisfied, where we ignore each edge incident to $v$ that is labeled $\portedge$, and for each edge with input label $\gadedge$ the edge constraints $C^E_{\gadget}$ of $\Psi_{\gadget}$ have to be satisfied.

    \begin{remark*}
    Here, as in the following descriptions, we will consider the labels defined above as a collection of several labels in the canonical way, e.g., each edge has three input labels, one each from $\Sigma^{\Pi}_{E,\inn}$, $\Sigma^{\gadget}_{E,\inn}$, and $\{\portedge, \gadedge\}$. Also, for simplicity, we will not explicitly mention which of the labels are relevant for the respective constraint if this is clear from the context. For instance, the (only) labels the above constraint for solving $\Psi_{\gadget}$ talks about (apart from the labels from $\{\portedge, \gadedge\}$ that determine which edges are considered for the constraint) are the input and output labels for $\Psi_{\gadget}$, i.e., the input labels from $\Sigma^{\smash\gadget}_{V,\inn}$, $\{\port_1, \dots, \port_{\Delta}, \noport\}$, $\Sigma^{\smash\gadget}_{E,\inn}$, and $\Sigma^{\smash\gadget}_{B,\inn}$, and the output labels from $\Sigma^{\smash\gadget}_{V,\out}$, $\Sigma^{\smash\gadget}_{E,\out}$, and $\Sigma^{\smash\gadget}_{B,\out}$.
    \end{remark*}
	\item\label{item:error} Each node $v$ has to be labeled $\lporterr_2$ if and only if $v$ has input label $\port_i$ for some $i$ and, either there is no incident edge labeled $\portedge$, or there are at least two incident edges labeled $\portedge$. Otherwise $v$ has to be labeled either $\lporterr_1$ or $\lnoporterr$.
	\item\label{item:niceedge} For each edge $e = \{ u, v \}$ with input label $\portedge$ the following holds: If $u$ and $v$ are labeled $\port_i$ and $\port_j$ for some $1 \leq i, j \leq \Delta$, respectively, and the output label $\in \Sigma^{\gadget}_{V, \out}$ of both $u$ and $v$ is $\lgadok$, then the output label $\in \{\lporterr_1, \lporterr_2, \lnoporterr\}$ of both $u$ and $v$ cannot be $\lporterr_1$; if $u$ is labeled $\port_i$ for some $i$ and at least one of $u$ and $v$ has input label $\noport$ or an output label from $L_{\lerr}$, then the output label $\in \{\lporterr_1, \lporterr_2, \lnoporterr\}$ of $u$ cannot be $\lnoporterr$.
	\item\label{item:nodelist} For each node $v$ with incident edges $e_1, \dots, e_k$, if at least one of $v, e_1, \dots, e_k, (v, e_1), \dots, (v, e_k)$ is assigned an output label from $L_{\lerr}$ and none of the node constraints mentioned above are violated, then the node constraint for $v$ is always satisfied, irrespective of the conditions below. If all of the mentioned elements of $V \cup E \cup B$ are assigned an output label from $\{ \lgadok, \epsilon \}$, then the following conditions have to be satisfied for $v$, where \[\ell^v_{\llist} = \bigl(S, \iota^V, \iota^E_1, \dots, \iota^E_{\Delta}, \iota^B_1, \dots, \iota^B_{\Delta}, o^V, o^E_1, \dots, o^E_{\Delta}, o^B_1, \dots, o^B_{\Delta}\bigr)\] denotes the $\Sigma_{\llist}$-part of the output label assigned to $v$:
	\begin{itemize}
		\item If $v$ is labeled $\port_i$ for some $1 \leq i \leq \Delta$, then the label $\port_i$ is an element of $S$ if and only if the output $\in \{\lporterr_1, \lporterr_2, \lnoporterr\}$ of $v$ is $\lnoporterr$.
		\item If $v$ is labeled $\port_1$, then $\iota^V$ is $v$'s input label from $\Sigma^{\Pi}_{V, \inn}$.
		\item If $v$ is labeled $\port_i$ for some $1 \leq i \leq \Delta$, and $\port_i \in S$, then for any
			incident edge $e$ labeled $\portedge$, the labels $\iota^E_i$ and $\iota^B_i$ coincide with $e$'s input label from $\Sigma^{\Pi}_{E, \inn}$ and $(v,e)$'s input label from $\Sigma^{\Pi}_{B, \inn}$, respectively.
		\item The output label of $v$ encodes a configuration that satisfies the node constraints from $C^{\Pi}_V$. More precisely, let $\alpha$ be the bijection that monotonically maps the elements of $\{ 1, \dots, |S|\}$ to the indices of the elements in $S$, and consider a (hypothetical) node $u$ of degree $|S|$ with incident edges $e'_1, \dots, e'_{|S|}$. Then labeling $u, e'_1, \dots, e'_{|S|}, (u,e'_1), \dots, (u,e'_{|S|})$ with input labels \[\iota^V, \iota^E_{\alpha(1)}, \dots, \iota^E_{\alpha(|S|)}, \iota^B_{\alpha(1)}, \dots, \iota^B_{\alpha(|S|)}\] and output labels \[o^V, o^E_{\alpha(1)}, \dots, o^E_{\alpha(|S|)}, o^B_{\alpha(1)}, \dots, o^B_{\alpha(|S|)},\] respectively, yields a correct node configuration at $u$ according to $C^{\Pi}_V$.
	\end{itemize}
	\item\label{item:edgelist} Similarly, for each edge $e = \{u, v\}$, if at least one of $u, v, e, (u,e), (v,e)$ is assigned an output label from $L_{\lerr}$ and none of the node constraints mentioned above are violated, then the edge constraint for $e$ is always satisfied, irrespective of the conditions below. If all of the mentioned elements of $V \cup E \cup B$ are assigned output labels from $\{ \lgadok, \epsilon \}$, then the following conditions have to be satisfied for $e$, where $\ell^u_{\llist}$ and $\ell^v_{\llist}$ denote the $\Sigma_{\llist}$-part of the output labels assigned to $u$ and $v$, respectively, and we use the above notation augmented with a superscript to indicate the respective node:
	\begin{itemize}
		\item If $e$ is labeled $\gadedge$, then $\ell^u_{\llist} = \ell^v_{\llist}$.
		\item If $e$ is labeled $\portedge$, and $u$ and $v$ are labeled $\port_i$ and $\port_j$ for some $1 \leq i, j \leq \Delta$, respectively, then $\iota^{E,u}_{\alpha(i)} = \iota^{E,v}_{\alpha(j)}$ and $o^{E,u}_{\alpha(i)} = o^{E,v}_{\alpha(j)}$, and, for a (hypothetical) edge $e' = \{ u', v' \}$, labeling $u', v', e', (u', e'), (v', e')$ with input labels \[\iota^{V,u}, \iota^{V,v}, \iota^{E,u}_{\alpha(i)}, \iota^{B,u}_{\alpha(i)}, \iota^{B,v}_{\alpha(j)}\] and output labels \[o^{V,u}, o^{V,v}, o^{E,u}_{\alpha(i)}, o^{B,u}_{\alpha(i)}, o^{B,v}_{\alpha(j)},\] respectively, yields a correct edge configuration at $e'$ according to $C^{\Pi}_E$.
	\end{itemize} 
\end{enumerate}

\paragraph{Informal description.}
\begin{itemize}
	\item \emph{Input labels.} Elements of $B$ can be intuitively seen as endpoints of an edge, thus we will refer to them as ``half-edges''. Each node, each edge, and each half-edge has an input for $\Pi$ and an input for $\Psi_{\gadget}$. Also, each node (resp.\ edge) may have a special label indicating if it is a port node (resp.\ edge).
	\item \emph{Output labels.} Each node must produce a tuple $(s, p, g)$. The labeling $g$ must be a valid output for $\Psi_{\gadget}$. The labeling $p$ is used to indicate the (in)correctness of the port connections. The labeling $s = (l,i,o)$ is the one that actually contains a solution for $\Pi$. First, $l$ contains a list of valid ports of the gadget. Then, $i$ contains a copy of all the inputs of the port nodes, as well as the inputs of their edges and half-edges, that is, everything that is needed to know the input of a virtual node. Finally, $o$ contains the output of the virtual node, described as node, edges and half-edges outputs. All these labels will be useful to check the validity of the output for $\Pi$ in a local manner.
	\item \emph{Constraints.} For each aforementioned constraint, we provide an informal description, by following the same order.
	\begin{enumerate}
		\item[1.] We require that outputs for $\Psi_{\gadget}$ do not cross gadget boundaries. Thus, we require that port edges and half-edges are labeled $\epsilon$, while everything else must actually contain outputs for $\Psi_{\gadget}$.
		\item[2.] Each connected component, given by removing port edges from the graph, must provide a valid solution for $\Psi_{\gadget}$.
		\item[3.--4.] Port nodes do not output errors only in the case in which they are connected to exactly one other port node, and both of them are in a correct gadget.
		\item[5.] Nodes claiming that the gadget is correct must:
		\begin{itemize}
			\item Produce a list of valid ports of the gadget.
			\item Copy the node input of $\port_1$, that will be treated as the input for the virtual node (this is an arbitrary choice, but since nodes may be provided with different inputs for $\Pi$, we need nodes to agree on some specific input for the virtual node).
			\item Copy edge and half-edge inputs of port nodes to the output.
			\item Produce outputs that are correct w.r.t.\ the constraints of $\Pi$.
		\end{itemize}
		\item [6.] On edges we first check that nodes of the same gadget are giving the same output. Then, port edges check that the edge constraints for $\Pi$ are satisfied on the virtual edges.
	\end{enumerate}
\end{itemize}

\subsection{Upper and lower bounds}
We now proceed by showing upper and lower bounds for the defined \nelcl{} $\Pi'$, which, together, will then imply Theorem~\ref{thm:pi_to_newpi}. Intuitively, in order to solve $\Pi'$, nodes can do the following. They start exploring the graph to see if they are in a valid gadget. If they see that their gadget is invalid, then they can produce a locally checkable proof of error. Otherwise, they need to solve the original problem, by first seeing which ports are connected to exactly one valid gadget (on all other ports they can output $\lporterr_1$ or $\lporterr_2$), and then simulating the algorithm for the original problem $\Pi$ on the graph obtained by contracting the valid gadgets to a node and ignoring invalid gadgets.

\subsubsection{Upper bound}
\begin{lemma}\label{lem:newpi_ub}
	Problem $\Pi'$ can be solved in $O(\tdet(\Pi,n) \cdot d(n))$ rounds deterministically, and in $O(\trand(\Pi,n) \cdot d(n))$ rounds randomized.
\end{lemma}
\begin{proof}
Let $\gadcheck$ be the algorithm guaranteed by Definition \ref{def:gadget}, able to produce a locally checkable proof of the (in)validity of the gadget or, equivalently, solving the \nelcl{} $\Psi_{\gadget}$.
Each node $v$ starts by executing $\gadcheck$ on the connected components of the subgraph obtained by ignoring edges labeled $\portedge$, which can be done in time $O(d(n))$ where $n$ denotes the number of nodes of the input graph.
For simplicity, we will refer to these connected components as \emph{gadgets}, where we say that a gadget is \emph{valid} if $\gadcheck$ returns label $\lgadok$ everywhere in the gadget, and \emph{invalid} if $\gadcheck$ returns at least one label from $L_{\lerr}$.
Node $v$ then outputs the labels returned by $\gadcheck$ on itself and the incident edges and elements of $B$, thereby providing the part of the output labels corresponding to $\Sigma^{\smash\gadget}_{V, \out}$, $\Sigma^{\smash\gadget}_{E, \out}$, and, $\Sigma^{\smash\gadget}_{B, \out}$, respectively, in the description of the output labels.
Since $\gadcheck$ solves $\Psi_{\gadget}$, this takes care of Constraint~\ref{item:check}; by outputting $\epsilon$ on all edges with input label $\portedge$ and all associated elements of $B$, we see that also Constraint~\ref{item:eps} is satisfied.

If $v$ is a node labeled $\noport$, then it outputs $\lnoporterr$.
If $v$ is labeled $\port_i$ for some $i$, then it gathers its constant-radius neighborhood and checks whether it is a ``valid" port:
If $v$ has no incident edge labeled $\portedge$ or at least two incident edges labeled $\portedge$, then it outputs $\lporterr_2$.
If $v$ has exactly one incident edge labeled $\portedge$, then it checks whether itself or the other endpoint $u$ of the edge is labeled $\noport$ or outputs an element of $L_{\lerr}$ after executing $\gadcheck$.
If one of the conditions is satisfied, then $v$ outputs $\lporterr_1$, otherwise it outputs $\lnoporterr$.
This takes care of Constraints~\ref{item:error} and \ref{item:niceedge}.

If a gadget is invalid, then by Constraints~\ref{item:nodelist} and \ref{item:edgelist}, the constraint for each node and edge in the gadget is satisfied, and we simply complete the outputs for all nodes in the gadget in an arbitrary way that conforms to the output label specifications.
Hence, what remains is to assign to each node $v$ in a valid gadget the $\Sigma_{\llist}$-part $\ell^v_{\llist} = \bigl(S, \iota^V, \iota^E_1, \dots, \iota^E_{\Delta}, \iota^B_1, \dots, \iota^B_{\Delta}, o^V, o^E_1, \dots, o^E_{\Delta}, o^B_1, \dots, o^B_{\Delta}\bigr)$ of the output label in a way that ensures that Constraints~\ref{item:nodelist} and \ref{item:edgelist} are satisfied.
This is the part where, intuitively, we solve the original problem $\Pi$ on the graph obtained by ignoring all invalid gadgets and contracting the valid gadgets to single nodes which are then connected by the edges labeled $\portedge$.
We proceed as follows, considering only nodes in valid gadgets.
Each node collects all input and hitherto produced output information contained in its gadget and the gadget's radius-$1$ neighborhood, and uses the obtained knowledge to determine the first part of $\ell^v_{\llist}$ by choosing $S, \iota^V, \iota^E_1, \dots, \iota^E_{\Delta}, \iota^B_1, \dots, \iota^B_{\Delta}$ in a way that conforms to Constraint~\ref{item:nodelist}.
The choices for the mentioned labels immediately follow from the conditions in Constraint~\ref{item:nodelist} (or can be freely chosen, for some labels).

For determining the second part of $\ell^v_{\llist}$ (corresponding to the actual outputs in the solution of~$\Pi$), each node solves the original problem $\Pi$ as follows:
\begin{itemize}
	\item If the aim is a deterministic algorithm for $\Pi'$, then gather the radius-$O(\tdet(\Pi,n) \cdot d(n))$ neighborhood; if the aim is a randomized algorithm, then gather the radius-$O(\trand(\Pi,n) \cdot d(n))$ neighborhood.
	\item Construct a (partial) virtual graph $H$ by contracting each valid gadget to a single node and deleting all nodes in invalid gadgets---note that this may result in a graph $H$ with parallel edges and/or self-loops, which is why, in our model, we allow graphs to contain these. For each virtual node $u$, assign port numbers from $1$ to $\deg(u)$ to the incident edges in the only way that respects the order of the indices of the gadget's $\port_i$ nodes the corresponding $\portedge$ edges are connected to.
	\item Assign, as identifier of a virtual node, the smallest id of its associated gadget.
	\item Compute $q = (q^V, q^E_1, \dots, q^E_{\deg(u)}, q^B_1, \dots, q^B_{\deg(u)})$, a valid solution for $\Pi$ for the current virtual node and its incident edges and elements of $B$, where the indices indicate the corresponding port for the respective edge or element of $B$.
\end{itemize}
Now, each node $v$ in a valid gadget transforms the output $q$ of the virtual node $u$ corresponding to the gadget into the desired tuple $(o^V, o^E_1, \dots, o^E_{\Delta}, o^B_1, \dots, o^B_{\Delta})$ as follows.
Recall the function $\alpha$ defined in Constraint~\ref{item:nodelist}, and set \[\bigl(o^V, o^E_{\alpha(1)}, \dots, o^E_{\alpha(\deg(u))}, o^B_{\alpha(1)}, \dots, o^B_{\alpha(\deg(u))}\bigr) = \bigl(q^V, q^E_1, \dots, q^E_{\deg(u)}, q^B_1, \dots, q^B_{\deg(u)}\bigr).\]
Note that, by construction $\deg(u) = |S|$.
Now it is straightforward (if somewhat cumbersome) to check that this completion of the output of $v$ satisfies the last bullet of Constraint~\ref{item:nodelist} and the first of Constraint~\ref{item:edgelist}.
The remaining second bullet of Constraint~\ref{item:edgelist} follows from the fact that the computed outputs $q$ form a valid solution to \nelcl{} $\Pi$.

We need to show that, given their radius-$O(\tdet(\Pi,n) \cdot d(n))$ neighborhood, resp.\ radius-$O(\trand(\Pi,n) \cdot d(n))$ neighborhood in the randomized case, nodes can actually find a valid solution for the original problem $\Pi$. To this end, we want to show that after collecting this neighborhood nodes can see up to a radius of at least $\tdet(\Pi,n)$, resp.\ $\trand(\Pi,n)$, in the virtual graph, and that any virtual graph has size at most $n$. This follows from the following observations:
\begin{itemize}
	\item In the worst case a gadget has diameter $O(d(n))$, where the worst case occurs if there is a single gadget containing all the nodes of the graph.
	\item In the worst case for the size of the virtual graph each gadget has just constant size. Even in this case the virtual graph has at most $n$ nodes.
\end{itemize}
Hence, each node can simulate a $\tdet(\Pi,n)$-round, resp. $\trand(\Pi,n)$-round, algorithm for $\Pi$ (whose existence is guaranteed by $\Pi$'s time complexity) on the virtual graph and thus find a valid solution for $\Pi$, as required.
Note that this simulated algorithm assumes that the input graph for $\Pi$ (i.e., the \emph{virtual} graph) has size $n$, which is an assumption that is consistent with the view of each node since we allow disconnected graphs (which is important in case a node sees the whole virtual graph, which is then interpreted as a connected component of an $n$-node graph).
It follows that, in the randomized case, the failure probability of our obtained algorithm for $\Pi'$ is upper bounded by the failure probability of the used algorithm for $\Pi$ since the algorithm for $\Pi'$ only fails if the algorithm for $\Pi$ would fail on an $n$-node graph that contains the virtual graph as a connected component.
In particular, the obtained randomized algorithm gives a correct output w.h.p.

Since the gathering process dominates the execution time, the time complexity of the obtained algorithm is $O(\tdet(\Pi,n) \cdot d(n))$ in the deterministic case, and $O(\trand(\Pi,n) \cdot d(n))$ in the randomized case.
Note that this algorithm works also on graphs containing self-loops and parallel edges.
\end{proof}

\subsubsection{Lower bound}
\begin{lemma}\label{lem:newpi_lb}
	Let $f\colon \NN \to \NN$ be a function such that, for each $x \in \NN$, we have $f(x) \leq x$  and there exists some $y \in \NN$ with $f(y) = x$.
	Solving problem $\Pi'$ requires $\Omega(\tdet(\Pi,f(n)) \cdot d(\frac{n}{f(n)}))$ rounds deterministically and $\Omega(\trand(\Pi,f(n)) \cdot d(\frac{n}{f(n)}))$ rounds randomized.
\end{lemma}
\begin{proof}
We start by proving the randomized lower bound.
For a contradiction, assume that there is a $o(\trand(\Pi,f(n)) \cdot d(\frac{n}{f(n)}))$-round randomized algorithm $\AAA$ that solves $\Pi'$ w.h.p.
Recall the $(d, \Delta)$-gadget family $\gadget$ used to define $\Pi'$.
Let $N$ be the largest integer with $N \leq n/f(n)$ such that there exists a gadget $G_N \in \gadget$ with $N$ nodes such that the pairwise distances between the ports of $G_N$ are all in $\Theta(d(N))$.
Let $H$ be an arbitrary graph with $f(n)$ nodes, and consider the padded graph $H' \in \gadget(H)$ obtained by choosing gadget $G_N$ for each node of $H$.
Let $H''$ be the $n$-node graph obtained by adding $n - N \cdot f(n)$ isolated nodes to $H'$.

Consider what happens if the nodes in $H$ simulate $\AAA$ on $H''$.
Since each node of $H$ has been expanded into a valid gadget, a valid solution for $\Pi'$ found by $\AAA$ on the subgraph $H'$ of $H''$ yields a valid solution for $\Pi$ on $H$, by the definition of problem $\Pi'$.
Hence, due to the properties of our function $f$, we have transformed $\AAA$ into an algorithm $\AAA'$ for $\Pi'$, and the failure probability of $\AAA'$ on graphs of size $f(n) \leq n$ is upper bounded by the failure probability of $\AAA$ on graphs of size $n$.
It follows that $\AAA'$ is correct w.h.p.
Moreover, in order to simulate $\AAA$, it is sufficient if each node of $H$ collects its radius-$o(\trand(\Pi,f(n)) \cdot d(\frac{n}{f(n)}))/\Theta(d(N))$, due to the runtime of $\AAA$ and the definition of $G_N$.
By Definition~\ref{def:gadget} and the definition of $N$, we have $N = \Theta(n/f(n))$; therefore, the runtime of $\AAA$ is $o(\trand(\Pi,f(n)))$.
This yields a contradiction to the definition of $\trand(\Pi,f(n))$ and proves the randomized lower bound.
The deterministic lower bound is proved analogously, the only difference being that we do not have to worry about failure probabilities.
\end{proof}
Theorem~\ref{thm:pi_to_newpi} now follows from Lemmas~\ref{lem:newpi_ub} and \ref{lem:newpi_lb}.

\section{\boldmath A \texorpdfstring{$(\log,\Delta)$}{(log, \textDelta)}-gadget family}\label{sec:thegadget}
In this section we present a $(\log,\Delta)$-gadget family, and prove that it satisfies the properties described in Definition \ref{def:gadget}. Hence, we will prove the following theorem.
\begin{theorem}\label{thm:gadget-family}
	There exists a $(\log ,\Delta)$-gadget family.
\end{theorem}

Informally, each gadget in the $(\log,\Delta)$-gadget family is composed by $\Delta$ sub-gadgets. Each sub-gadget is a complete binary tree where we add horizontal edges, creating a path that traverses nodes of the same level. The bottom right node of each sub-gadget is a \emph{port} (see Figure \ref{fig:subgadget}). Then, we add a node, that we call \emph{center}, and connect it to the root of each sub-gadget  (see Figure \ref{fig:gadget}). Also, we add constant-size input labels to the gadget to make its structure locally checkable.

As stated in Definition \ref{def:gadget}, $\Psi_\gadget$ must be a \nelcl. For the sake of readability, we will define $\Psi$ as a constant radius checkable \lcl. Then, we will show how to modify it and obtain a \nelcl{} $\Psi_\gadget$.

\subsection{Sub-gadget}
For any parameter $h$, it is possible to construct sub-gadgets of height $h$. Let $(\ell_u,x_u)$ be the \emph{coordinates} of a node $u$ of the sub-gadget. For any node $u$, it holds $0\le\ell_u < h$ and $0\le x_u < 2^{\ell_u}$. Let $u$ and $v$ be two nodes with coordinates $(\ell_u,x_u)$ and $(\ell_v,x_v)$ respectively, such that $\ell_v \le \ell_u$ and $x_v \le x_u$.  There is an edge between $u$ and $v$ if and only if:
\begin{itemize}[noitemsep]
	\item $(\ell_v,x_v)=(\ell_u-1, \lfloor \frac{x_u}{2} \rfloor)$, or
	\item $(\ell_u, x_u)=(\ell_v, x_v+1)$.
\end{itemize}

\paragraph{Sub-gadget labels.}
We make a sub-gadget locally checkable by adding constant-size labels in the following way. First of all, each node $u=(\ell_u,x_u)$ has labels:
\begin{itemize}[noitemsep]
	\item $\lindex_i$, where $1\le i \le \Delta$;
	\item $\port_i$,  where $1\le i \le\Delta$, if $\ell_u=h-1$ and $x_u=2^{\ell_u}-1$.
\end{itemize} 
Moreover, each edge $e=\{ u,v\}$ has a label on both endpoints, $L_u(e)$ and $L_v(e)$. Each label $L_u(e)$ is chosen as follows:
\begin{itemize}[noitemsep]
	\item $L_u(e)=\parent$ if $(\ell_v,x_v)=(\ell_u-1, \lfloor \frac{x_u}{2} \rfloor)$;
	\item $L_u(e)=\lright$ if $(\ell_v, x_v)=(\ell_u, x_u+1)$;
	\item $L_u(e)=\lleft$ if $(\ell_v, x_v)=(\ell_u, x_u-1)$;
	\item $L_u(e)=\lchild$ if $(\ell_v, x_v)=(\ell_u + 1,2x_u)$;
	\item $L_u(e)=\rchild$ if $(\ell_v,x_v)=(\ell_u + 1, 2x_u+1)$.
\end{itemize}
See Figure \ref{fig:subgadget} for an example of a sub-gadget.

\begin{figure}[t]
	\centering
	\includegraphics[scale=0.7]{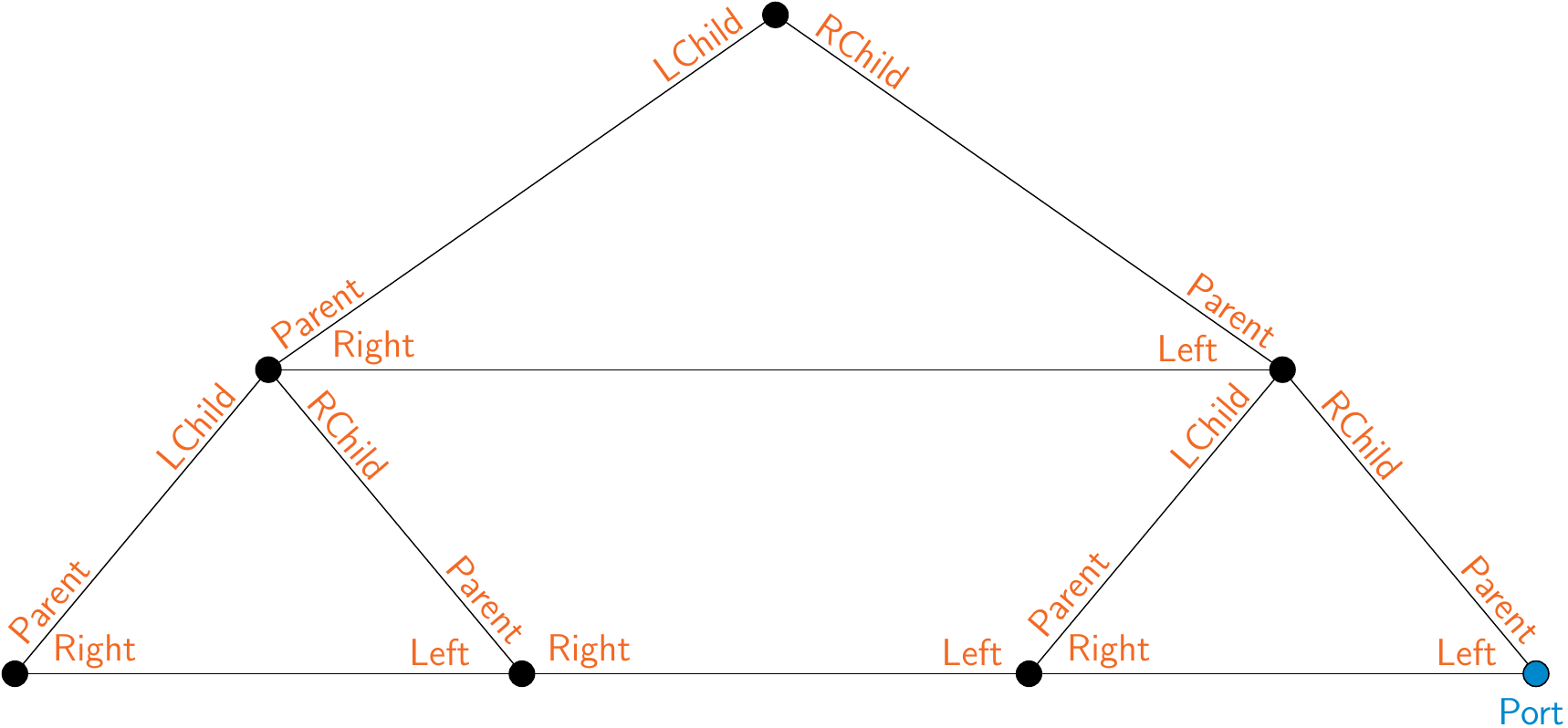}
	\caption{An example of a sub-gadget and its input labeling.}\label{fig:subgadget}
\end{figure}

\begin{figure}[t]
    \centering
    \includegraphics[scale=0.6]{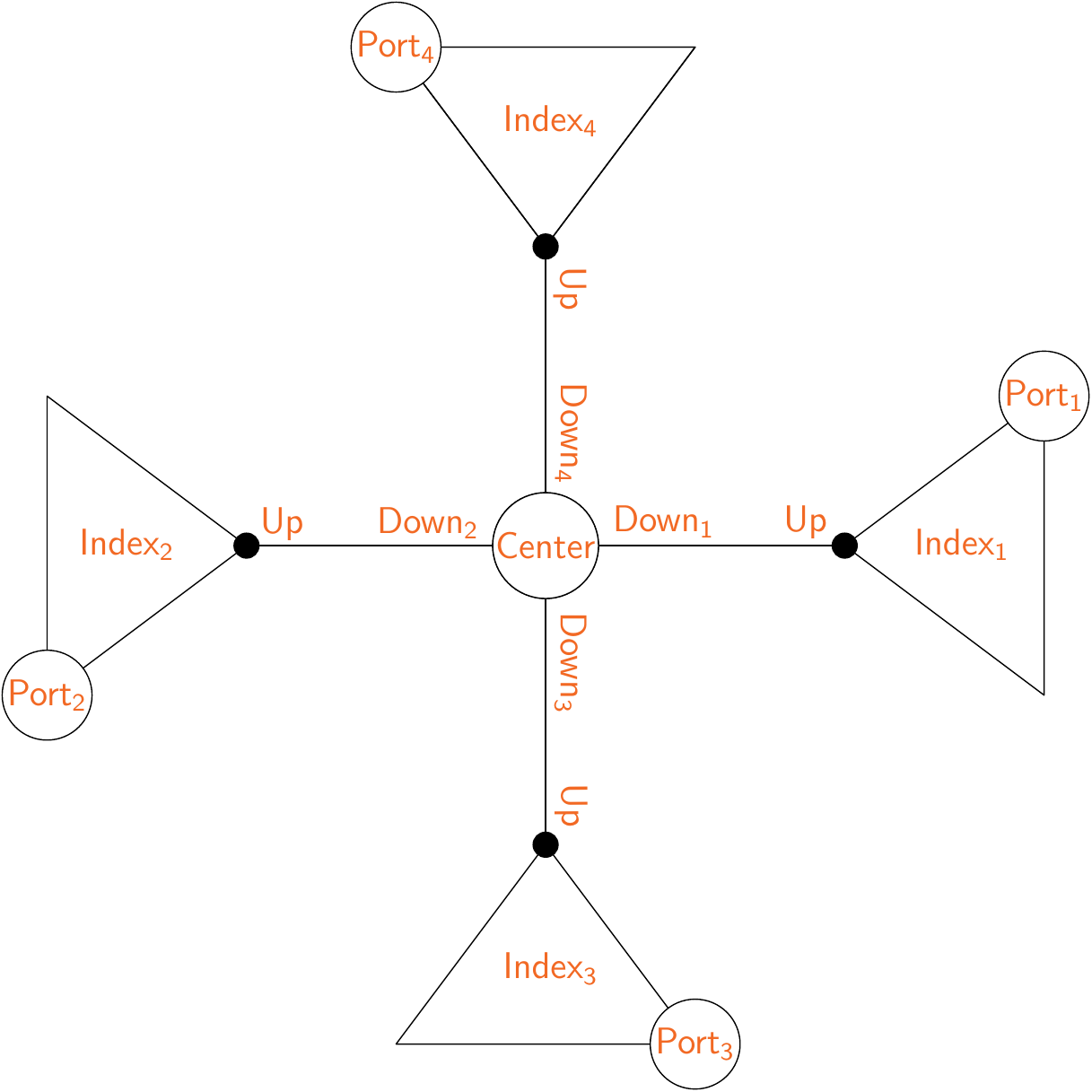}
    \caption{An example of a gadget and its input labeling.}\label{fig:gadget}
\end{figure}

\subsection{Local checkability of a sub-gadget}\label{subsec:SG_checkability}
Let $L_1, L_2, \dots, L_k$ be $k$ labels. We denote by $u(L_1,L_2,\dots,L_k)$ the node reached from $u$ by following edges labeled $L_1,L_2,\dots,L_k$. Each node $u$ of the sub-gadget checks the following local constraints.
\begin{enumerate}
	\item Each node $u$ checks the following to guarantee some basic properties:
	\begin{enumerate}
		\item there are no self loops or parallel edges;\label{cons:no_selfloop}
		\item for any two incident edges $e,e'$, $L_u(e)\neq L_u(e')$;
		\item it must be labeled $\lindex_i$ for some $i$, and its neighbors must be labeled $\lindex_i$ as well;
		\item if $u$ is labeled $\port_i$ and $\lindex_j$, then $i=j$.\label{cons:correct_port_label}
	\end{enumerate}
	\item Each node $u$ checks the following to guarantee a correct internal structure of the sub-gadget: 
	\begin{enumerate}
		\item for each edge $e=\{u,v\}$, if $L_u(e)=\lleft$ then $L_v(e)=\lright$, and vice versa;\label{cons:left-right}
		\item for each edge $e=\{u,v\}$, if $L_u(e)=\parent$ then $L_v(e)=\rchild$ or $L_v(e)=\lchild$, and vice versa;
		\item $u(\lchild,\lright,\parent)=u$, if the path exists;\label{cons:triangle}
		\item $u(\lright,\lchild,\lleft,\parent)=u$, if the path exists.\label{cons:cycles}
	\end{enumerate}
	\item Each node $u$ checks the following to guarantee the correct boundaries of the sub-gadget:
	\begin{enumerate}
		\item $u$ does not have an incident edge labeled $\lright$ if and only if neither $u(\parent)$ does, if it exists;\label{cons:right boudary}
		\item $u$ does not have an incident edge labeled $\lleft$, if and only if neither $u(\parent)$ does, if it exists;\label{cons:left boudary}
		\item if $u$ does not have an incident edge with label $\lright$ and it has an incident edge $e=\{u,v\}$ labeled $L_e(u)=\parent$, then $L_e(v)=\rchild$;
		\item if $u$ does not have an incident edge labeled $\lleft$ and it has an incident edge $e=\{u,v\}$ labeled $L_e(u)=\parent$, then $L_e(v)=\lchild$;
		\item if $u$ does not have incident edges labeled $\lright$ and $\lleft$, then it is the root of the sub-gadget and it has only two incident edges with labels $\lchild$ and $\rchild$;\label{cons:root}
		\item $u$ has an incident edge labeled $\rchild$ if and only if it also has an incident edge labeled $\lchild$;\label{cons:2 children}
		\item if $u$ does not have incident edges with labels $L_u(e)=\lchild$ or $L_u(e)=\rchild$, then neither does $u(\lleft)$ and $u(\lright)$ (if they exist);\label{cons:bottom boundary}
		\item $u$ is labeled $\port_i$ if and only if it does not have incident edges labeled $\lright$, $\lchild$, and $\rchild$.
	\end{enumerate}
\end{enumerate}
If the above constraints are satisfied, we say that the sub-gadget has a valid structure.

\paragraph{Correctness.}
We want to show two things: a valid sub-gadget satisfies all the above constrains, and, any graph that satisfies the above constraints is a valid sub-gadget. It is clear that the first property holds. In order to prove the second property, we will proceed as follows. First we will show that a graph that satisfies the above constraints must have a node that does not contain incident edges labeled $\rchild$ or $\lchild$.  Then, assuming we have such a node in the graph, we prove that it is a valid sub-gadget. 

\begin{lemma}\label{lem:SG_correctness}
	Let $G$ be a graph with $n$ nodes that satisfy the local constraints of a sub-gadget, then $G$ is a valid sub-gadget.
\end{lemma}
\begin{proof}
	By constraints \ref{cons:no_selfloop}--\ref{cons:correct_port_label}, each node satisfies the basic properties of a valid sub-gadget, such as the consistency of the labels. Constraints \ref{cons:left-right}--\ref{cons:cycles} ensure that the internal structure of the graph looks like a valid sub-gadget. Assume, by contradiction, that all nodes in $G$ have an incident edge with label $\rchild$. By constraint \ref{cons:2 children}, all nodes have also an incident edge labeled $\lchild$.  In order for each node to have $2$ children, such that no node has two incident edges labeled $\parent$, we need to have $2n$ nodes in $G$, which is a contradiction. This means that there exists a node $u$ in $G$ that is not a parent. By constraint \ref{cons:bottom boundary}, we ensure that also nodes $u(\lleft)$ and $u(\lright)$, if they exist, do not have incident edges with labels $\rchild$ or $\lchild$, ensuring that $G$ has a bottom boundary, as desired.
	
	Suppose that all nodes $u$ that do not have incident edges labeled $\rchild$ or $\lchild$ have also an incident edge labeled $\lright$. Notice that $u(\lright)$ cannot end in a node that has  incident edges with labels $\rchild$ or $\lchild$, since it would contradict constraint \ref{cons:bottom boundary}. Also, by constraint \ref{cons:no_selfloop}, self loops are not allowed. Hence, by constraint \ref{cons:left-right} every node in the bottom boundary must have incident edges with labels $\lright$ and $\lleft$. This means that the bottom boundary wraps around horizontally, forming a cycle. By constraints \ref{cons:right boudary} and \ref{cons:left boudary}, the graph $G$ will continue to wrap around horizontally, and since the internal structure is valid, the size of these cycles must halve each time, reaching a node $u$ (the root) that satisfies $u(\lright)=u$, which contradicts constraint \ref{cons:no_selfloop}.
	
	Hence, among nodes that have no incident edges labeled with $\rchild$ or $\lchild$, there must exist a node $u$ such that it does not have an incident edge $\lright$. This implies that there must exist also a node $v$ that does not have an incident edge labeled $\lleft$. Constraints  \ref{cons:left-right}--\ref{cons:cycles} and constraints \ref{cons:right boudary}--\ref{cons:root} ensure that $G$ has left and right boundaries according to the ones of a valid sub-gadget. Putting all together, we conclude that $G$ has the structure of a valid sub-gadget.
\end{proof}

\subsection{Gadget}\label{subsec:G_checkability}
A gadget is composed by $\Delta$ sub-gadgets, and the root of each sub-gadget is connected to a central node, labeled $\lcenter$. Let $u$ be a central node, then each edge $e=\{u,v\}$ has the following labels:
\begin{itemize}[noitemsep]
	\item let $\lindex_i$ be the label of node $v$, then $L_u(e)=\ldown_i$;
	\item $L_v(e)=\lup$.
\end{itemize}
See Figure \ref{fig:gadget} for an example of a gadget.

\paragraph{Local checkability.} In addition to the constraints described for a sub-gadget, each node $u$ checks also the following local constraints:
\begin{enumerate}
	\item if $u$ does not have an incident edge labeled $\parent$, it checks that it has exactly one neighbor labeled $\lcenter$;\label{cons:existence of center}
	\item if $u$ is labeled with $\lcenter$, it checks that:
	\begin{enumerate}
		\item $u$ is connected to exactly $\Delta$ nodes (roots of sub-gadgets);\label{cons:delta SG}
		\item for any edge $\{u,v\}$, let $\lindex_i$ be the label of node $v$, then $L_u(e)=\ldown_i$;
		\item for any edge $\{u,v\}$, $L_v(e)=\lup$;
		\item let $v$ and $w$ be neighbors of $u$, if $v$ is labeled $\lindex_i$ and $w$ is labeled $\lindex_j$, then $i\neq j$. \label{cons:diff index}
	\end{enumerate}
\end{enumerate}
If the above constraints are satisfied, we say that a gadget is valid.

\paragraph{Correctness.} It is easy to see that a gadget satisfies all the above constrains. We want to show that any graph that satisfies the above constraints is a gadget (notice that we have already shown the local checkability of a sub-gadget and its correctness, so we will assume that we are dealing with valid sub-gadgets).

\begin{lemma}\label{lem:G_correctness}
	Let $G$ be a graph with $n$ nodes that satisfies the local constraints of a gadget, then $G$ is a valid gadget.
\end{lemma}
\begin{proof}
	In Lemma \ref{lem:SG_correctness} we have shown the correctness of a single sub-gadget. We still need to show that there cannot be edges among different sub-gadgets. This is ensured by constraint \ref{cons:diff index} and the constraint that, for each  node, label $\lindex$ must be the same as the one of its neighbors. Also, constraint \ref{cons:existence of center} guarantees the existence of a node labeled $\lcenter$. By constraints \ref{cons:delta SG}--\ref{cons:diff index}, we have that the central node is correctly connected to the $\Delta$ sub-gadgets, ensuring that the graph has the structure of a valid gadget.
\end{proof}

\subsection{\boldmath \lcl{} problem \texorpdfstring{$\Psi$}{\textPsi}}\label{subsec:invalid_gadget}
We now define a constant radius checkable \lcl{} problem $\Psi$, where either:
\begin{itemize}[noitemsep]
	\item all nodes output $\lok$, or
	\item all nodes output a (possibly different) error label.
\end{itemize}
On one hand, if the structure of a gadget is invalid, nodes must be able to prove that there is an error.  On the other hand, if the structure of the gadget is valid, then nodes must not be able to claim that there is an error. Notice that we allow nodes to output $\lok$ even if the gadget is invalid. Moreover, if the gadget is invalid, we show that it is possible to prove that there is an error in $O(\log n)$ rounds. More precisely, the possible error output labels of the nodes are the following:
\begin{itemize}[noitemsep]
	\item an error label $\lerror$;
	\item an error pointer label in $\{\lright, \lleft, \parent, \rchild,\lup,\ldown_i\}$.
\end{itemize}
The error output labels must satisfy the following locally checkable constraints.
\begin{enumerate}
	\item A node outputs either $\lerror$ or exactly one error pointer.
	\item A node outputs $\lerror$ if and only if the input labeling does not satisfy the local constraints given in Sections \ref{subsec:SG_checkability} and \ref{subsec:G_checkability}.
	\item Let $u$ be a node that outputs an error pointer, then the following hold.
	\begin{enumerate}
		\item If the error pointer is $\lright$, then $u(\lright)$ outputs $\lerror$ or an error pointer $\lright$.\label{cons:error_right}
		\item If the error pointer is $\lleft$, then $u(\lleft)$ outputs $\lerror$ or an error pointer $\lleft$.\label{cons:error_left}
		\item If the error pointer is $\parent$, then $u(\parent)$ outputs $\lerror$ or an error pointer in $\{\parent,\lleft, \lright, \lup\}$.\label{cons:error_parent_pointer}
		\item If the error pointer is $\rchild$, then $u(\rchild)$ outputs $\lerror$ or an error pointer in $\{\rchild, \lright, \lleft\}$.\label{cons:error_Rchild}
		\item If the error pointer is $\lup$ and $u$ has label $\lindex_i$, then $u(\lup)$ outputs $\lerror$ or an error pointer $\ldown_j$, where $j\neq i$.\label{cons:error_down(j)}
		\item If the error pointer is $\ldown_i$, then $u(\ldown_i)$ outputs $\lerror$ or an error pointer $\rchild$.
	\end{enumerate}
\end{enumerate}

\begin{lemma}
	There does not exist an algorithm that, on a valid gadget, outputs error labels that satisfy the local constraints at all nodes.
\end{lemma}
\begin{proof}
	We show that if a gadget is valid, then it is not possible to output error labels such that the above constraints are locally satisfied at all nodes. Hence, suppose we have a gadget that has a valid structure where nodes produce error labels. First of all, since the gadget is valid, then there is no node outputting $\lerror$, as it would violate the local constraints described above. Hence, nodes can only output error pointers in $\{\lright, \lleft, \parent, \rchild,\lup,\ldown_i\}$. There are two cases: 
	\begin{enumerate}
		\item all nodes of all sub-gadgets point towards the node labeled $\lcenter$ (i.e., the central node is a sink);
		\item the central node points towards the root of a sub-gadget.
	\end{enumerate}  
	In the first case, all error chains will end up at the central node that cannot output $\lerror$, which violates the error constraints. So, suppose that the node labeled $\lcenter$ produces an error pointer, that is, it outputs $\ldown_i$ for some $1\le i \le \Delta$. This means that we just need to show that we cannot cheat inside a sub-gadget. Hence, consider a sub-gadget, and suppose that the central node points towards the root of this sub-gadget. Notice that the root node cannot output $\lup$, since it would violate the error pointer constraints. Hence, all nodes $u$ of the sub-gadget must output an error pointer in $\{\lright, \lleft, \parent, \rchild\}$. The following hold.
	\begin{itemize}
		\item If the error pointer is $\lright$, then, according to the error label specifications, $u(\lright)$ can only output $\lright$ or $\lerror$. Since the structure is valid, no node will output $\lerror$, hence this chain will propagate until it reaches a node in the sub-gadget that does not have an incident edge labeled $\lright$, which contradicts constraint \ref{cons:error_right}. The case when the error pointer is $\lleft$ is analogous.
		\item If the error pointer is $\parent$, then, according to the error label specifications, the chain either reaches the root of the sub-gadget, or, at some point, a node in the chain outputs either $\lleft$ or $\lright$. The latter case is handled above, while in the former case, according to the error label constraints, the root should output $\lup$ (since it cannot output $\lerror$). But the root cannot point $\lup$ since, in that case, the root and the central node would point to each other, contradicting constraint \ref{cons:error_down(j)}.
		\item If the error pointer is $\rchild$, then, according to the specifications, $u(\rchild)$ can only output an error label in $\{\rchild, \lright, \lleft \}$ (since it cannot output $\lerror$). This chain cannot end at a node $v$ that has been reached traversing only edges labeled $\rchild$, as in that case $v$ should output $\lerror$, which is not allowed. So, at some point of the chain, there is a node that outputs $\lright$ or $\lleft$, and, as shown above, this would lead to a violation of the constraints. \qedhere
	\end{itemize}
\end{proof}

\subsection{Upper bound}
We show an algorithm $\gadcheck$ that satisfies Definition \ref{def:gadget}, i.e., given an upper bound $n$ on the size of the graph, in case of a valid gadget, it outputs $\lok$ at every node, while in case of an invalid gadget, $\gadcheck$ is able to provide error labels in $O(\log n)$ rounds satisfying the error constraints. Let $u(L_1^{i_1},L_2^{i_2},\dots,L_k^{i_k})$ be the node reached from $u$ by following label $L_1$ $i_1$ times, label $L_2$ $i_2$ times, and so on. The algorithm is as follows.

\begin{enumerate}
	\item Each node gathers its constant-radius neighborhood and checks if the constraints in Sections \ref{subsec:SG_checkability} and \ref{subsec:G_checkability} are satisfied. 
	\item If the specifications of a valid gadget are not satisfied in a node's constant-radius neighborhood, then the node outputs $\lerror$.
	\item If the constraints are satisfied in the constant-radius neighborhood of node $u$, then $u$ gathers its $O(\log n)$-radius neighborhood.
	\item If a node does not see any error in its $O(\log n)$-radius neighborhood, then it outputs $\lok$.
	\item If a node $u$ labeled $\lcenter$ sees an error in its $O(\log n)$-radius neighborhood, then $u$ outputs $\ldown_i$ if the error can be reached by following the labels $u(\ldown_i^1,\rchild^{i_1}, \lright^{i_2})$ or $u(\ldown_i^1,\rchild^{j_1}, \lleft^{j_2})$, where $i_1,i_2,j_1,j_2\ge 0$, breaking ties by choosing the smallest label $\ldown_i$ that satisfies the above.
	\item If a node $u$ not labeled $\lcenter$ sees an error in its $O(\log n)$-radius neighborhood, then it acts according to the following specifications, that a node checks in order. 
	\begin{enumerate}
		\item if there exists an error that can be reached following $u(\lright^i)$, where $i\ge 1$, $u$ outputs $\lright$;
		\item if there exists an error that can be reached following $u(\lleft^i)$, where $i\ge 1$, $u$ outputs $\lleft$;
		\item if there exists an error that can be reached following the labels $u(\parent^{i_1}, \lright^{i_2})$ or $u(\parent^{j_1}, \lleft^{j_2})$, where $i_1,j_1\ge 1$ and $i_2,j_2\ge 0$, then $u$ outputs $\parent$;
		\item if there exists an error that can be reached following the labels $u(\rchild^{i_1}, \lright^{i_2})$ or $u(\rchild^{j_1}, \lleft^{j_2})$, where $i_1,j_1\ge 1$ and $i_2,j_2\ge 0$, then $u$ outputs $\rchild$;
		\item if none of the above happens, it means that node $u$ is in a valid sub-gadget and the error is outside this sub-gadget; in this case, node $u$ outputs $\parent$ if it has an incident edge with that label, otherwise $u$ outputs $\lup$.
	\end{enumerate}
\end{enumerate}

\begin{lemma}
	The time complexity of algorithm $A$ is $O(\log n)$.
\end{lemma}
\begin{proof}
	By definition, algorithm $A$ runs in $O(\log n)$ rounds. Also, notice that a valid sub-gadget is a complete binary tree-like structure, and there is a node (the central one) that is connected to the root of each sub-gadget. This means that, in $O(\log n)$ rounds, a node either sees an error or it sees all the gadget. So, if a gadget looks like a valid structure from the perspective of a node, it means that it is globally a valid gadget, and in that case every node outputs $\lok$. We need to show that, in case of an invalid gadget, the algorithm produces valid error labels, that is, the constraints described in Section \ref{subsec:invalid_gadget} are satisfied at each node that outputs $\lerror$ or an error pointer.
	
	Consider a gadget that has an invalid structure. According to the specifications of algorithm $A$, a node outputs $\lerror$ if and only if the constraints that determine a valid gadget are not satisfied in its constant-radius neighborhood, as desired. So, consider a node $u$ that outputs an error pointer, and let us denote with $w$ a node that outputs $\lerror$ that is at distance $O(\log n)$ from $u$. 
	\begin{enumerate}
		\item If there is a path that connects node $u$ to a node $w$ by using only labels $\lright$, then $u$ will output $\lright$. This holds for every node between $u$ and $w$, resulting in an error chain that traverses only edges labeled $\lright$, and ends at a node that witnesses an error. This error chain behaves according to constraint \ref{cons:error_right} in Section \ref{subsec:invalid_gadget}.
		\item If the above does not apply, then $u$ checks if there is a path connecting $u$ and a node $w$ using only labels $\lleft$, and if that is the case, $u$ outputs $\lleft$. Similarly as before, this chain satisfies constraint \ref{cons:error_left}  in Section \ref{subsec:invalid_gadget}.
		\item If the above cases do not hold, then node $u$ checks if there is a path that connects $u$ and a node $w$ using $i\ge 1$ times the label $\parent$, followed by $j\ge 0$ times $\lright$ or $k\ge 0$ times $\lleft$. If that is the case, $u$ outputs $\parent$. This error chain, either ends at an ancestor of $u$ that outputs $\lerror$, or, at some point in the error chain there will be an ancestor $v$ of $u$ that reaches a node $w$ following only $\lright$, or only $\lleft$ labels, ending the error chain at a node outputting $\lerror$. This error chain behaves according to constraint \ref{cons:error_parent_pointer} in Section \ref{subsec:invalid_gadget}.
		\item If the above cases do not apply, then node $u$ checks whether there is a path connecting $u$ and $w$ using $i\ge 1$ times the label $\rchild$, followed by $j\ge 0$ times $\lright$ or $k\ge 0$ times $\lleft$, and if that is the case, $u$ outputs $\rchild$. Again, this chain ends either at a descendant that outputs $\lerror$, or, at some point there will be a node $v$ that can reach $w$ following only $\lright$, or only $\lleft$ labels, behaving according to constraint \ref{cons:error_Rchild} in Section \ref{subsec:invalid_gadget}.\label{case_Rchild}
		\item If a node $u$ not having incident edges labeled $\ldown_i$ (i.e., $u$ is not a central node) passes all the above checks, it means that it cannot reach a node $w$ that outputs $\lerror$ without traversing an edge labeled $\lup$. Hence $u$ is a node of a valid sub-gadget, and the error is somewhere else. In this case, all nodes in the valid sub-gadget point towards the central node, that is, a node outputs $\parent$ if it has an incident edge with such a label, otherwise it outputs $\lup$, satisfying constraints \ref{cons:error_parent_pointer} and \ref{cons:error_down(j)} in Section \ref{subsec:invalid_gadget}.
		\item The central node $u$ can only output an error pointer $\ldown_i$, breaking ties by choosing the label $\ldown_i$ having smallest index, such that an error can be reached by following $u(\ldown_i^1, \rchild^{i_1},\lright^{i_2})$ or $u(\ldown_i^1, \rchild^{j_1},\lleft^{j_2})$ ($i_1,i_2,j_1,j_2\ge 0$). Notice that the central node cannot point to the root of a valid sub-gadget, since no error can be reached from it following labels $\rchild$, or $\lright$, or $\lleft$. It is easy to see that if $u(\ldown_i)$ outputs $\lerror$, then the constraints are satisfied. Otherwise, for the same reasoning used in point \ref{case_Rchild}, we conclude that the error chain behaves according to constraint \ref{cons:error_parent_pointer} in Section \ref{subsec:invalid_gadget}.
	\end{enumerate}
Notice that at least one of the above cases apply when a node sees an error. Hence, a node outputs an error or an error pointer if and only if the structure is invalid.
	\end{proof}

\subsection{Node-edge checkability of the gadget}
We now show how to modify problem $\Psi$ defined in Section \ref{subsec:invalid_gadget} and obtain a \nelcl{} $\Psi_\gadget$. For this purpose, we must show that the validity of the output of nodes is checkable according to some node and edge constraints. 

First of all, it is easy to see that all error pointers can be expressed as node and edge constraints. For example, consider the constraint ``If the error pointer is $\lright$, then $u(\lright)$ outputs $\lerror$ or an error pointer $\lright$''. We can encode this constraint by requiring the following.
\begin{itemize}
	\item Node constraints: a node outputs consistently on all its incident edges, i.e., if a node outputs $\lright$ on one incident edge, it must output $\lright$ on all other incident edges.
	\item Edge constraints: for an edge $e=\{u,v\}$ input labeled $L_u(e)=\lright$ and $L_v(e)=\lleft$, if $u$'s side output label is $\lright$, $v$'s side output label is either $\lright$ or $\lerror$.
\end{itemize}
All other error pointer constraints can be encoded in a similar way.

Handling the label $\lerror$ requires more care. In fact, as it is defined, $\Psi$ allows nodes to output $\lerror$ if they locally see an inconsistency in the gadget structure. While this output is checkable by exploring a constant radius neighborhood, it may not necessarily be node-edge checkable. The cases that require more attention are constraints \ref{cons:no_selfloop}, \ref{cons:triangle}, and \ref{cons:cycles} of Section \ref{subsec:SG_checkability}, since all the others can be handled similarly as we did with error pointers.

\paragraph{Handling constraint \ref{cons:no_selfloop}.}
According to constraint \ref{cons:no_selfloop}, we need to allow nodes to output $\lerror$ if there are self loops or parallel edges. It seems not possible to prove the presence of a self loop or parallel edges in the node-edge formalism. Thus, instead of requiring the nodes to prove the presence of a self loop or a parallel edge, we require the input labeling to prove the \emph{absence} of self loops and parallel edges. For this purpose, as input label for the nodes of the gadget, we add a distance-$2$ coloring with $O(\Delta^2)$ colors. It is trivial to see that self loops do not admit a proper coloring of the graph. Also, in case of parallel edges, there are two edges connected to the same neighbor, violating the constraint of being a proper distance-$2$ coloring. In order to make everything node-edge checkable, we require that the color of each node is replicated on all its incident edges as well (recall that each edge may be input labeled differently on each side).

We now show how to handle the case in which there is a node $v$ such that it has two incident edges ending on nodes with the same color $c$ (this includes the parallel edges case). In this case, node $v$ is allowed to output $\lerror$ by specifying a color $c$ and two incident edges that connect $v$ to neighbors of color $c$. The constraints are as follows (see Figure \ref{fig:parallel-edges} for an example):
\begin{itemize}[noitemsep]
	\item Node constraints: two edges must be specified, outputting on each of them the same color $c$.
	\item Edge constraints: if an edge is output labeled $c$ on one side, it is input labeled with the same color on the other side.
\end{itemize} 
Other distance-$2$ color errors can be handled similarly.

\begin{figure}[t]
	\centering
	\includegraphics[scale=1.2]{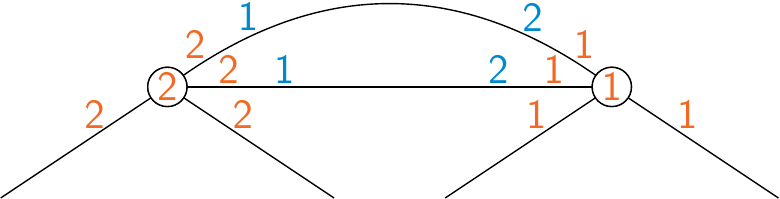}
	\caption{An example of node-edge checkability in the case of parallel edges: the input label is shown in orange, while the output label is shown in blue.}\label{fig:parallel-edges}
\end{figure}

Notice that, having a distance-$2$ coloring in input does not change the complexity of $\gadcheck$, that is, it still remains not possible to claim that there is an error in a valid instance.

\paragraph{Handling constraint \ref{cons:cycles}.}
We now show how to handle constraint \ref{cons:cycles}, as constraint \ref{cons:triangle} can be handled similarly. Constraint \ref{cons:cycles} allows a node $u$ to output $\lerror$ if $u(\lright,\lchild,\lleft,\parent) \neq u$ (if the path exists). We show how $u$ and its neighbors can prove the error claim in a node-edge checkable manner. The idea is that nodes $u$, $u(\lright)$, $u(\lright,\lchild)$, $u(\lright,\lchild,\lleft)$, and $u(\lright,\lchild,\lleft,\parent)$, can label themselves with a chain of labels such as $A,B,C,D,E$. In a valid graph, this is not possible because $u$ would have both labels $A$ and $E$. These labels can be node-edge checked. For example:
\begin{itemize}[noitemsep]
	\item Node constraints: if a node is labeled $A$, it must have label $A$ on all its incident edges
	\item Edge constraints: if an edge is labeled $A$ and $\lright$ on one side, it needs to be labeled $B$ on the other side.
\end{itemize}
The node-edge constraints are similar for nodes that output labels different from $A$ in the chain. See Figure \ref{fig:cycle} for an example.
\begin{figure}[t]
	\centering
	\includegraphics[scale=0.8]{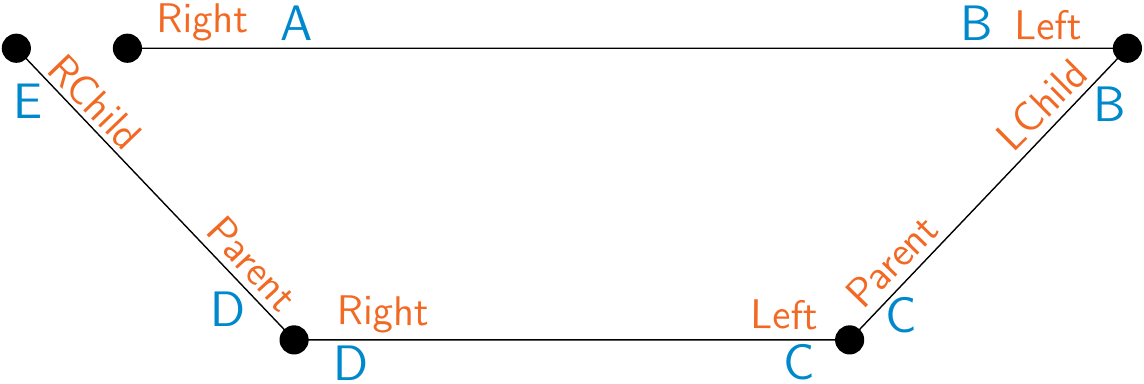}
	\caption{An example of node-edge checkability in the case where constraint \ref{cons:cycles} is not satisfied: the input label is shown in orange, while the output label is shown in blue.}\label{fig:cycle}
\end{figure}

Finally, we would like that \emph{all} nodes that do not satisfy locally constraint \ref{cons:cycles} are able to produce such a proof. The problem is that we do not currently allow overlapping chains, as it would allow to produce errors in valid graphs. This can be solved by requiring nodes to properly color the chains. That is, each chain has a color from a large enough palette, and nodes can participate to different chains by tagging each label with the color of the chain. For example, a node could output $\{(c_1,A),(c_2,E)\}$. Since we are in bounded degree graphs, and since chains have constant length, this requires an additive term of $O(\log^* n)$ rounds on the running time of algorithm $\gadcheck$, thus the complexity of $\gadcheck$ does not change. 

\subsection{Validity of the gadget family}
In this section we presented a gadget family $\gadget$ and proved that it satisfies some properties. Now we show that it is a $(\log, \Delta)$-gadget family, i.e., it satisfies the properties in Definition \ref{def:gadget}, proving Theorem \ref{thm:gadget-family}. Consider a graph $G=(V,E)\in \gadget$, where $|V|=n$. Then $G$ satisfies the following. 
\begin{itemize}
	\item The number of nodes is trivially $n$.
	\item Each of the $\Delta$ sub-gadgets has a port node labeled $\port_i$, where $1\le i \le \Delta$.
	\item Each sub-gadget has a complete binary tree-like structure, hence its diameter is $O(\log n)$. Since the root of each sub-gadget is connected to the central node, the diameter of the gadget and the pairwise distance between the ports is $O(\log n)$.
\end{itemize}
The above observations show that, according to Definition \ref{def:gadget}, any $n$-node gadget $G\in \gadget$ is an $(n,O(\log n))_\Delta$-gadget. Moreover, we showed that checking whether a graph $G$ is contained in $\gadget$ is a \nelcl{} solvable in $O(\log n)$ communication rounds, given an upper bound $n$ on the size of the network. In order to show that $\gadget$ is really a $(\log, \Delta)$-gadget family, we still need to show that, for any $n\in \NN$, there exists a $G \in \gadget$ with $\Theta(n)$ nodes such that the pairwise distances between the ports are all in $\Theta(\log n)$. This is indeed satisfied by those gadgets $G\in \gadget$ having all $\Delta$ sub-gadgets of the same size. 

\section{Putting things together}\label{sec:newlcl}

In this section, we combine our findings of Sections~\ref{sec:padded-lcl} and \ref{sec:thegadget} in order to provide a family of \lcl{} problems where randomization helps, but only subexponentially.
We obtain this family by starting from the sinkless orientation problem and recursively applying Theorem~\ref{thm:pi_to_newpi}.

More precisely, we define a family of \lcl{}s $\Pi^i$ having deterministic complexity $\Theta(\log^i n)$ and randomized complexity $\Theta(\log^{i-1} n\log\log n)$, for any constant $i = 1, 2, \dotsc$. 
The base case $i=1$ is given by the sinkless orientation problem, for which deterministic and randomized tight bounds of $\Theta(\log n)$ and $\Theta(\log\log n)$, respectively, are known \cite{ghaffari17distributed,chang16exponential,Brandt2016}.
Note that these bounds also hold in our setting, where we allow self-loops, parallel edges, and disconnected graphs.
The problem $\Pi^{i+1}$ is obtained by applying Theorem~\ref{thm:pi_to_newpi} to $\Pi^i$ and the $(\log, \Delta)$-gadget family whose existence we proved in Theorem~\ref{thm:gadget-family}, where we set $f(x) := \lfloor \sqrt{x} \rfloor$.

From Theorem \ref{thm:pi_to_newpi}, we know that, given a problem $\Pi$ of time complexity $T(\Pi,n)$ and using a $(\log,\Delta)$-gadget family and the specified function $f(x)$, we obtain a new \lcl{} $\Pi'$ of complexity $O\bigl(T(\Pi,n) \cdot \log n\bigr)$ and $\Omega\bigl(T(\Pi,\sqrt{n}) \cdot \log(\sqrt{n})\bigr)$, for both the deterministic and the randomized case.
Starting from problem $\Pi^i$ having deterministic and randomized complexities $\Theta(\log^i n)$ and $\Theta(\log^{i-1} n\log\log n)$, we obtain that the problem $\Pi^{i+1}$ has:
\begin{itemize}
	\item Deterministic complexity $O\bigl(\log^i n \log n\bigr)$ and $\Omega\bigl(\log^i(\sqrt{n}) \log (\sqrt{n})\bigr)$, obtaining a tight complexity of $\Theta(\log^{i+1}n)$.
	\item Randomized complexity $O\bigl(\log^{i-1} n\log\log n \log n\bigr)$ and $\Omega\bigl(\log^{i-1}(\sqrt{n})\log\log(\sqrt{n}) \log (\sqrt{n})\bigr)$, obtaining a tight complexity of $\Theta(\log^i n \log\log n)$.
\end{itemize}
From the above observations, we obtain the following theorem. 
\begin{theorem}
	There exist \lcl{} problems with deterministic complexity $\Theta(\log^i n)$ and randomized complexity $\Theta(\log^{i-1} n \log\log n)$, for any $i = 1, 2, \dotsc$.
\end{theorem}

\section*{Acknowledgments}

We thank Fabian Kuhn for discussions related to network decompositions.
This work was supported in part by the Academy of Finland, Grant 285721.

\urlstyle{same}
\bibliographystyle{plainnat}
\bibliography{padding}

\end{document}